\documentclass{svjour3}
\usepackage{amsfonts}
\usepackage{amsmath} %,amsthm
\usepackage{graphicx}
\usepackage{booktabs} 
\usepackage{graphbox} 
\usepackage{multirow}
\usepackage{adjustbox}
\usepackage{rotating}
\usepackage{natbib}
\usepackage{tikz}
\usepackage{hyperref}

\usepackage{algorithmic}
\usepackage{algorithm}

\usepackage{adjustbox}
\usepackage{hyperref}
\usepackage{array}
\newcommand{\nt}[1]{\begin{tabular}[c]{@{}l@{}}#1\end{tabular}}
\newcolumntype{L}[1]{>{\raggedright\let\newline\\\arraybackslash\hspace{0pt}}m{#1}}
\def\X{$\times$}

\smartqed

\if@abbrvbib
\else
\fi

\bibpunct{(}{)}{;}{a}{,}{,}

\def\makeheadbox{\relax}

\title{Discovery Data Topology with the Closure Structure. Theoretical and practical aspects}
\titlerunning{Discovery Data Topology with the Closure Structure}
\author{Tatiana Makhalova \and Aleksey Buzmakov \and Sergei O. Kuznetsov \and Amedeo Napoli}
\institute{Tatiana Makhalova \at
    Universit\'{e} de Lorraine, CNRS, Inria, LORIA, F-54000 Nancy, France \\
    \email{tatiana.makhalova@inria.fr}
    \and
    Aleksey Buzmakov \at 
    National Research University Higher School of Economics, Perm, Russia\\
    \email{avbuzmakov@hse.ru}
    \and
    Sergei O. Kuznetsov \at 
    National Research University Higher School of Economics, Moscow, Russia\\
    \email{skuznetsov@hse.ru}
    \and 
    Amedeo Napoli \at
    Universit\'{e} de Lorraine, CNRS, Inria, LORIA, F-54000 Nancy, France \\
    \email{amedeo.napoli@loria.fr}
}
\date{}

\begin{document}

\maketitle

\begin{abstract}
 In this paper, we are revisiting pattern mining and especially itemset mining, which allows one to analyze binary datasets in searching for interesting and meaningful association rules and respective itemsets in an unsupervised way.
While a summarization of a dataset based on a set of patterns does not provide a general and satisfying view over a dataset, we introduce a concise representation --the closure structure-- based on closed itemsets and their minimum generators, for capturing the intrinsic content of a dataset.
The closure structure allows one to understand the topology of the dataset in the whole and the inherent complexity of the data.
We propose a formalization of the closure structure in terms of Formal Concept Analysis, which is well adapted to study this data topology.
We present and demonstrate theoretical results, and as well, practical results using the GDPM algorithm.
GDPM is rather unique in its functionality as it returns a characterization of the topology of a dataset in terms of complexity levels, highlighting the diversity and the distribution of the itemsets.
Finally a series of experiments shows how GDPM can be practically used and what can be expected from the output.

\end{abstract}

%%%%%%%%%%%%%%%%%%%%%%%%%%%%%%%%%%%%%%%%%%%%%%%%%%%%%%%%%%%%%%%%%%%%%%%%%%%
% Amedeo, the new introduction is here...
%
% ------------- Questions -------------
% 1. We use different punctuation throughout the paper: what would you prefer: "i.e., <text>" or "i.e. <text>"?
% -------------------------------------

\section{Introduction}

In this paper, we are revisiting pattern mining and especially itemset mining, although this research area did not attract much attention in the last years.
Actually, itemset mining allows to analyze binary datasets in searching for interesting and meaningful association rules and respective itemsets in an unsupervised way~\citep{han2011data}.
There are two main existing approaches to itemset mining (IM) which are based on
(i) an exhaustive enumeration of itemsets satisfying a set of predefined constraints followed by an assessment of the resulting itemsets --thanks to an interestingness measure for example-- for building with the most interesting itemsets the most concise representation of the data,
or (ii) sampling small sets of itemsets that meet a given set of constraints.

In the first approach, we can distinguish miners based on frequent itemsets.
Over the last decades, a large variety of algorithms for mining frequent itemsets was proposed, e.g., join-based algorithms~\citep{agrawal1994fast} which rely on a level-wise exploration, where $(k+1)$-itemsets are derived from $k$-itemsets, vertical exploration~\citep{shenoy2000turbo}, projection-based techniques~\citep{han2000mining}, etc.
This intensive development of algorithms of frequent itemset mining (FIM) is mainly due to the fact that frequent itemsets have an intuitive interpretation, namely ``the most interesting itemsets are frequently occurring in the data''.
Moreover, frequency is an antimonotonic measure w.r.t. the size of the itemset and this allows to enumerate very efficiently all frequent itemsets and only them.
However, frequent itemset mining shows several drawbacks as frequent itemsets are rather not (necessarily) interesting and in general very redundant.
In addition, when the frequency threshold decreases, the problems of combinatorial explosion and itemset redundancy become more acute.

Focusing on closed itemsets~\citep{pasquier1999discovering,zaki1998theoretical} allows for a significant reduction of the number of itemsets by replacing a whole class of itemsets having the same support with the largest one, i.e. the corresponding closed itemset.
Nowadays, there exist very efficient algorithms for computing frequent closed itemsets~\citep{hu2017alpine,uno2005lcm}.
Even for a low frequency threshold, they are able to efficiently generate an exponential number of closed itemsets.
However, the efficient generation of closed itemsets solves the problem of the exponential explosion of itemsets only partially since the main difficulties appear after, when the generated itemsets are processed.

In the second approach, i.e. an alternative to the exhaustive enumeration of itemsets, more recent algorithms are based on ``sampling'' \citep{dzyuba2017flexible} and on a gradual search for itemsets according to an interestingness measure or a set of constraints \citep{smets2012slim}.
For example, Sofia is an algorithm which starts from one itemset using derivation rules and outputting a set of itemsets organized within a lattice and satisfying the constraints related to the derivation rules \citep{buzmakov2015fast}.
These algorithms usually output a rather small set of itemsets while they may provide only an approximate solution.

Considering both approaches, we can remark that, even if they use quite different techniques, i.e. exhaustive enumeration on the one hand and sampling on the other hand, these approaches rely on the same assumption, namely that the intrinsic structure of the dataset under study can be understood by means of selecting itemsets w.r.t. a subjective interestingness measure or a set of constraints.
Then, in each approach a particular set of itemsets is returned, which offers a ``multifaceted view'' about the data, but does not provide an ``intrinsic structure'' underlying the data.
This is precisely such a structure that we introduce and study here after.

In this paper\footnote{%
This paper is a revised and extended version of \citep{mkn2020}.
}
we propose a method and an algorithm for revealing and understanding the intrinsic structure of a dataset, called the \emph{closure structure}, which shows the distribution of the itemsets in the data in terms of frequency, and in addition supports an interpretation of the content of a dataset.
This closure structure is based on the ``closed itemsets'' and minimum elements of their ``equivalence classes'' in an unsupervised setting, and may be computed independently of any interestingness measure or set of constraints.
Actually, the closed itemsets are the base of what we call the ``topology of the dataset'', replacing the traditional open sets in mathematical topology.
Closed itemsets play the same role as open sets and are used to cover the dataset in a specific way, giving rise to the closure structure, which is based on a partition of levels, induced by the set of equivalence classes related to these closed itemsets.

The closure structure and its levels provide a representation of the complexity of the content of a dataset.
We propose a formalization of the closure structure in terms of Formal Concept Analysis \citep{ganter1999formal}, which is well adapted to the of study closed itemsets, equivalence classes, and data topology.
We present and demonstrate theoretical results about the closure structures, among which a characterization of the closure structure and complexity results about itemset mining, and as well a lower bound on the size of the pattern space --based on a lattice-- which is a good indicator of the complexity of the dataset.

%This research work is not only theoretical but also practical, as w
We introduce the GDPM algorithm for ``Gradual Discovery in Pattern Mining'' which computes the closure structure of a dataset.
Given a dataset, the GDPM algorithm returns a characterization of the levels constituting the closure structure of this dataset.
We study the complexity of the algorithm and compare it with related algorithms.
We also show in a series of experiments how GDPM can be used and what can be expected.
For example, it is possible to determine which are the levels where the itemsets are most diverse, how the frequency is distributed among the levels, and when a search for interesting itemsets can be stopped without loss.
Then we propose a comparative study on the intrinsic complexity and the topology and of a collection of public domain datasets.
To the best of our knowledge, GDPM is rather unique in its functionality as it returns a characterization of the topology of a dataset in terms of the closure structure, allowing to guide the mining of the dataset.

The paper has the following structure.
In Section~\ref{sec:basics} we recall the basic notions.
In Section~\ref{sec:closure_structure} we introduce the closure structure, discuss its properties and describe the GDPM algorithm for its computing.
Section~\ref{sec:experiments} contains the results of the experiments.
In Section~\ref{sec:conclusion} we conclude and give directions of future work.

%%%%%%%%%%%%%%%%%%%%%%%%%%%%%%%%%%%%%%%%%%%%%%%%%%%%%%%%%%%%%%%%%%%%%%%%%%%

\section{Related work}\label{sec:related_work}

In Data Mining and Knowledge Discovery computing itemsets is of main importance since itemsets are used to compute (i) a succinct representation of a dataset by small set of interesting and non-redundant patterns, (ii) a small set of non-redundant high-quality association rules. 

There are several approaches to compute itemsets: (i) an exhaustive enumeration of all itemsets followed by a selection of those satisfying imposed constraints~\citep{fpm14interestingness}, (ii) a gradual enumeration of itemsets guided by an objective (or by constraints)~\citep{smets2012slim}, (iii) mining top-$k$ itemsets w.r.t. constraints~\citep{mampaey2012summarizing}, (iv)~sampling a subset of itemsets w.r.t. a probability distribution that conforms to an interestingness measure
~\citep{boley2011direct,boley2012linear}. To reduce redundancy when enumerating itemsets the search space can be shrunk to \emph{closed itemsets}, i.e., the maximal itemsets among those that are associated with a given set of objects (support). Meanwhile, in sampling techniques with rejection this restriction of arbitrary itemsets to closed ones may increase the sampling time~\citep{boley2010formal}.

Generally, two type of data can be defined: (i) ``deterministic'', where the object-attribute relations reflect non-probabilistic properties and where the noise rate is low, and (ii) ``probabilisic'', where the relations between objects and attributes are of stochastic nature. For example, deterministic data include the profiles of users with features such as education, gender, marital status, number of children, etc. By contrast, probabilistic data may describe basic item lists, where each item is present or absent with a certain probability (since one could forget about an item or to buy an item spontaneously).

For probabilistic data, enumerating closed itemsets may be infeasible in practice, since the number of closed itemsets may grow exponentially fast w.r.t. the noise rate~\citep{klimushkin2010approaches}. However, for ``deterministic'' data it might be very important to have a complete set of itemsets that meet user-defined constraints or to get a set of implications (i.e., association rules with confidence 1). In the last decades, substantial efforts in different directions have been made to establish a theoretical foundation of closed itemset mining within several theoretical frameworks: Formal Concept Analysis~\citep{ganter1999formal}, set systems~\citep{boley2007efficient,boley2007mlg,boley2010listing} and Data Mining~\citep{pasquier1999discovering,pasquier1999efficient,PasquierTBSL05}. 

A main challenge in closed itemset mining is computational complexity. The number of closed itemsets can be exponential in the dataset size.  In~\citep{kuznetsov1993fast} an efficient depth-first search algorithm, called CbO, for computing closed itemsets with polynomial delay was proposed. To avoid enumerating the same itemsets several times CbO relies on a canonicity test. Loosely speaking, a closed itemset $I_1$ passes the canonicity set if an itemset $I_2$ that generates $I_1$ is lexicographically smaller (we introduce this notion in Section~\ref{ssec:gdpm}) than $I_1$ itself. The execution trace of CbO is a spanning tree of the lattice composed of all closed itemsets ordered by set-inclusion. Further, other successor algorithms were developed within the FCA framework~\citep{kuznetsov2002comparing}. One of the most efficient algorithm from the CbO family is In-Close~\citep{andrews2014partial}, where the canonicity is leveraged to compute a partial closure, i.e., checking the closure on a subset of attributes that potentially may not pass the canonicity test. The latter allows to reduce the running time. 

%In~\citep{boros2002complexity} the authors showed that (frequent) closed itemsets can be generated in incremental polynomial time, while all minimal infrequent itemsets (for a fixed frequency threshold) can be enumerated in incremental quasi-polynomial time. 

Later, Uno and Arimura proposed the LCM algorithm~\citep{uno2004efficient}, which computes all frequent closed itemsets with polynomial-delay and polyno\-mial-space complexity and is quite efficient in practice~\citep{uno2004lcm}. 
However, LCM requires a frequency threshold. The successor of LCM, Alpine algorithm~\citep{hu2017alpine}, does not have this drawback. It computes (closed) itemsets of decreasing frequency gradually and ensures completeness of the results, i.e., being interrupted anytime it returns a complete set of frequent (closed) itemset for a certain frequency threshold. 
Recently, in~\citep{janostik2020lcm} it was shown that LCM is basically CbO with some 
implementation tricks ensuring to LCM a very good performance, e.g., ``occurrence deliver'' (storing the attribute extensions in arraylists to avoid multiple traversal of a dataset), usage of ``conditional databases'' (reducing the dataset size in recursive calls by removing ``irrelevant'' objects and attributes and merging repetitive objects). 

In~\citep{krajca2009comparison} the authors study how different data structures, e.g., bitarrays, sorted linked lists, hash tables and others, affect the execution time. In the experiments, it was shown that binary search trees and linked lists are appropriate data structures for large sparse datasets, while bitarrays provide better execution time for small or dense datasets.

The theoretical aspects of closed itemset computing were studied in parallel within the framework of \emph{set systems}, where some elements commonly used in FCA were re-discovered. For example, in~\citep{boley2007efficient} the authors consider inductive generators and prove that every closed itemset (except $\emptyset$) has an inductive generator. Roughly speaking, it means that for each closed itemset $D$ there exists a closed itemset $D'$ and an item $e \in D \setminus D'$ such that the closure of $D'\cup \{e\}$ is equal to $D$. It follows directly from this definition that inductive generators are parents of the closed itemsets in the spanning tree used by the CbO algorithm. Later, in~\citep{arimura2009polynomial} a polynomial-delay and 
polynomial-space algorithm, called CloGenDFS, was proposed. As CbO, this algorithm avoids the explicit duplication test for itemsets. The execution tree used by CloGenDFS (called a {family tree}) is quite similar to that used by CbO in the bottom-up strategy. %The difference between them is that CbO builds the tree top-down, i.e., from the closure of $\emptyset$, while CloGenDFS builds the tree bottom-up. 

The results described above are related to \emph{computational aspects} of closed itemset mining. However, they do not address the problem of discovering an \textit{intrinsic structure} underlying the set of closed itemsets and more generally the dataset itself. In a seminal work related to this direction~\citep{pasquier1999discovering,stumme2002computing}, it was proposed to consider minimal (by set-inclusion) itemsets among all itemsets that are associated with the same set of objects. These sets are used to find a non-redundant set of association rules and implications (rules with confidence 1) that summarizes the data. In this work we consider minimum (by size) itemsets and use them to characterize this intrinsic structure and thus the data complexity.

The minimum implication base~\citep{ganter1999formal} --also known as Du\-quenne-Guigues basis~\citep{guigues1986familles}-- is one first way of characterizing the ``complexity'' of the set of closed itemsets w.r.t. implications. This basis represents the minimal set of implications sufficient to infer all other implications using the Armstrong rules~\citep{armstrong1974dependency}. Computing a minimal basis of implications is of particular importance in database theory, because there exists a one-to-one correspondence between the functional dependencies in an original database and implications computed on an associated database~\citep{ganter1999formal,codocedo2014semantic,baixeries2018characterizing}. 

Besides computing the minimal set of implications, there exists a range of approaches to computing non-redundant (not necessarily minimal in the sense of inference) set of implications or association rules. For example, in~\citep{zaki2000generating} the authors propose an approach to compute a set of \emph{general} implications and association rules based on closed itemsets. General rules have the smallest antecedent and consequent among the equivalent ones. In~\citep{bastide2000mining} an alternative approach for computing non-redundant implications and association rules was proposed. Each rule in this approach has the smallest antecedent and the largest consequent among the equivalent ones. And these rules are considered as ``good rules'', especially for interpretation purposes.

All approaches mentioned above rely implicitly or explicitly on an intrinsic structure, however this structure was rarely if ever formalized and its properties have not yet been studied in depth.

In this paper, we define a level-wise structure for representing the ``intrinsic structure'' of a dataset w.r.t. closed itemsets, that we call ``\emph{closure structure}''. We define the ``complexity'' of closed itemsets w.r.t. the size of minimal/minimum itemset common to a set of itemsets equivalent to the closed itemset itself. The levels of the closure structure are composed of closed itemsets having the same complexity. We also discuss computational issues, as well as theoretical and empirical properties of the closure structure.  

\section{Basic notions} \label{sec:basics}

In this study we rely on closed itemsets since (i)~a closed itemset is the unique maximal itemset that represents all (maybe numerous) itemsets that cover the same set of objects (same image in terms of Data Mining), (ii) a closed itemset provides a \emph{lossless representation} of these itemsets, and (iii) a closed itemset provides a concise representation of implications and association rules~\citep{pasquier1999efficient}. For dealing with closed itemsets, we use the framework of FCA~\citep{ganter1999formal}, an applied branch of the lattice theory related to Data Mining and Machine Learning. Some properties of the closure structure are obtained using theoretical results from FCA and are considered in the next sections. Let us recall the main notions used in this paper.

We deal with binary datasets which are represented as formal contexts throughout the paper. A \textit{formal context} is a triple $\mathbb{K}=\left ( G,M,I \right )$, where $G$ is a set of objects, $M$ is a set of attributes and $I\subseteq G\times M$ is a relation called incidence relation, i.e., $\left ( g, m \right )\in I $ if object $g$ has attribute $m$. 

We also may consider numerical data represented as a many-valued context $\mathbb{K}_{num} =  (G, S, W, I_{num})$, where the fact $(g,s,w) \in I_{num}$ means that object $g$ has value $w$ for attribute $s$. A many-valued context can be transformed into a binary context in several ways.
In this study, for the sake of simplicity, we replace an interval of values or a single value with a binary attribute.
An example of this transformation is given in Table~\ref{tab:many_valued_context}.
For example, binary attributes $a$-$c$ correspond to intervals of real-valued attribute $s_1$, binary attributes $d$-$f$ correspond to intervals of integer-valued attribute $s_2$, and attributes $g$, $e$ corresponds to integer values of $s_5$, namely  4, 5, respectively. In both numerical and the corresponding binarized data, each object is described by the same number of attributes, e.g., in the considered example, each object has 3 numerical and 3 binary attributes.

\begin{table}
    \caption{A multi-valued context (left), its representation used for binarization (middle), and the corresponding binary context (right)}
    \label{tab:many_valued_context}
    \begin{minipage}[b][][t]{.25\textwidth}
        \setlength{\tabcolsep}{2.5pt}
        \begin{tabular}{l|lll} \toprule
             &$s_1$&$s_2$&$s_3$\\ \midrule 
            $g_1$ & 1.3 & 6 & 4 \\
            $g_2$ & 2.1 & 2 & 4 \\
            $g_3$ & 2.5 & 5 & 5 \\
            $g_4$ & 1.8 & 1 & 5 \\
            $g_5$ & 3.3 & 3 & 5 \\
            $g_6$ & 1.6 & 4 & 4 \\\bottomrule
        \end{tabular}
    \end{minipage}
    \begin{minipage}[b][][t]{.27\textwidth}
        \setlength{\tabcolsep}{2.pt}
        \begin{tabular}{llll}\toprule
            &$s_1$&$s_2$&$s_3$ \\\midrule 
            $g_1$ & [1,2) & [5,7) & 4 \\
            $g_2$ & [2,3) & [1,3) & 4 \\
            $g_3$ & [2,3)& [5,7) & 5 \\
            $g_4$& [1,2) & [1,3) & 5 \\
            $g_5$& [3,4)& [3,5) & 5 \\
            $g_6$& [1,2) & [3,5) & 4\\\bottomrule
        \end{tabular}
    \end{minipage}
    \begin{minipage}[b][][t]{.4\textwidth}
        \setlength{\tabcolsep}{2.pt}
        \begin{tabular}{l|lll|lll|ll} \toprule
             &\rotatebox{90}{$s_1 \in [1,2)$}&\rotatebox{90}{$s_1 \in [2,3)$}&\rotatebox{90}{$s_1 \in [3,4)$}&\rotatebox{90}{$s_2 \in [1,3)$}&\rotatebox{90}{$s_2 \in [3,5)$}&\rotatebox{90}{$s_2 \in [5,7)$}&\rotatebox{90}{$s_3 = 4$}&\rotatebox{90}{$s_3 = 5$}\\\midrule 
             &$a$&$b$&$c$&$d$&$e$&$f$&$g$&$i$\\ \midrule 
            $g_1$ &$a$&  &  &  &  &$f$&$g$&  \\
            $g_2$ &  &$b$&  &$d$&  &  &$g$&  \\
            $g_3$ &  &$b$&  &  &  &$f$&  &$i$\\
            $g_4$ &$a$&  &  &$d$&  &  &  &$i$\\
            $g_5$ &  &  &$c$&  &$e$&  &  &$i$\\
            $g_6$ &$a$&  &  &  &$e$&  &$g$& \\\bottomrule
        \end{tabular}
    \end{minipage}

\end{table}

Two derivation operators $\left ( \cdot \right )' $ are defined for $A\subseteq G$ and $B \subseteq M$ as follows:
$$A' = \left \{ m \in M \mid \forall  g \in A : gIm  \right \}, \;
B' = \left \{ g \in G \mid  \forall  m \in B : gIm  \right \}.$$

Intuitively, $A'$ is the set of common attributes to objects in $A$, while $B'$ is the set of objects which have all attributes in $B$. %The double application of  $\left ( \cdot \right )'$  is a closure operator, i.e., $\left ( \cdot \right )''$ is extensive, idempotent and monotone. 

Sets $A \subseteq G$, $B \subseteq M$, such that $A = A''$ and $B = B''$, are closed sets. For $A \subseteq G$, $B \subseteq M$, the pair $(A,B)$ such that $A'=B$ and $B'=A$, is called a \textit{formal concept}, then $A$ and $B$ are closed sets and called \textit{extent} and \textit{intent}, respectively.
In Data Mining~\citep{pasquier1999efficient}, an attribute is also called an item, an \textit{intent} corresponds to a \textit{closed itemset} or a \textit{closed pattern}, and the extent is also called the \textit{image} of the itemset.
The number of items in an itemset $B$ is called the \emph{size} or \emph{cardinality} of $B$.

A partial order $\leq$ is defined over the set of concepts as follows: $\left(A,B\right)\leq\left(C,D\right)$ iff $A \subseteq C$, or dually $D \subseteq B$. A pair $\left(A,B\right)$ is a subconcept of $\left(C,D\right)$, while $\left(C,D\right)$ is a superconcept of $\left(A,B\right)$. With respect to this partial order, the set of all formal concepts forms a complete lattice $\mathcal{L}$ called the \textit{concept lattice} of the formal context $(G,M,I)$. The size of the concept lattice is bounded from above by $O(2^{\min(|G|, |M|)})$.

In this paper we deal with closed itemsets, i.e., intents of formal concepts. The whole set of closed itemsets for a context $(G,M,I)$ we denote by $\mathcal{C}$. We distinguish maximal/minimal and maximum/minimum itemsets. Let $\mathcal{S}$ be a set of itemsets. We call an itemset $X \in \mathcal{S}$ maximal w.r.t. $\mathcal{S}$ if there is no other itemset in $\mathcal{S}$ that includes $X$. Minimal itemset is defined dually. An itemset is called maximum in $\mathcal{S}$ if it has the largest \emph{size} or \emph{cardinality}, i.e., consists of the maximal number of items. Minimum itemset is defined dually.

Given an itemset $X$, its equivalence class $Equiv(X)$ is the set of itemsets with the same extent, i.e., $Equiv(X) = \{Y \mid Y \subseteq X, X' = Y'\}$.
A \textit{closed itemset} $X''$ is the maximum (largest cardinality) itemset in the equivalence class $Equiv(X)$. Thus, the concept lattice concisely represents the whole pattern search space, i.e., all possible combinations of items, where itemsets having the same support are encapsulated in the corresponding closed itemsets~\citep{pasquier1999discovering}.

\begin{definition}
A \emph{key} $X \in Equiv(B)$ is a minimal generator in $Equiv(B)$,  i.e., for every $Y \subset X$ one has $Y' \neq X' = B'$. We denote the set of keys (\textit{key set}) of $B$ by $Key(B)$. 
\label{def:key}
\end{definition}
%An itemset $X\in Equiv(B)$ is a \emph{key} or a \emph{minimal generator} of a closed itemset $B$ if for every $Y \subset X$ one has $Y' \neq X' = B'$. We denote a set of keys (\textit{key set}) of $B$ by $Key(B)$.

Hence, every proper subset of a key is a key that have a different closure~\citep{stumme2002computing}.
An equivalence class is \textit{trivial} if it consists of only one closed itemset. 

\begin{definition}
An itemset $X \in Key(B)$ is called a \emph{passkey} if it has the smallest size among all keys in $Key(B)$. We denote the set of passkeys by $pKey(B) \subseteq Key(B)$. For a closed itemset $B$ the \emph{passkey set} is given by $pKey(B) = \{ X \mid X \in Key(B), |X| = min_{Y \in Key(B)}|Y|\}$. 
\label{def:min_key}
\end{definition}

In an equivalence class there can be several passkeys, but only one closed itemset.
All passkeys have the same size in an equivalence class.
The size of the closed itemset is maximum among all itemsets in this equivalence class. 

\noindent\textbf{Example.} Let us consider the dataset in Fig.~\ref{fig:example}. The corresponding concept lattice is given on the right. Among all 16 closed itemsets only five have non-trivial equivalence classes, namely $ade$, $abc$, $bef$, $bdef$ and $abcdef$. Their keys and passkeys are listed in Table~\ref{tab:keys}. The remaining itemsets have trivial equivalence classes $Equiv(B) = \{B\}$. Moreover, every closed itemset $B$ with a trivial equivalence class has $|B|$ immediate subsets that are closed itemsets. For example, itemset $acd$ has a trivial equivalence class, and all its immediate subsets $ac$, $ad$ and $cd$ are closed. In Section~\ref{ssec:ideal} (Proposition 1) we analyze this property.

\begin{figure}
    \centering
	\begin{minipage}{.25\textwidth}
	\centering
	\setlength{\tabcolsep}{2.pt}
	\begin{tabular}{l|llllll}
			\toprule
			$g_1$&$a$&$b$&$c$& &&\\
     		$g_2$&$a$&   &$c$&$d$&   & \\
			$g_3$&$a$&   &   &$d$&$e$& \\
			$g_4$&   &   &$c$&$d$&   & \\
			$g_5$&   &$b$&   &$d$&$e$&$f$ \\
			$g_6$&   &$b$&   &   &$e$&$f$ \\\bottomrule
	\end{tabular}
	\end{minipage}
	\begin{minipage}{.65\textwidth}
	    \input{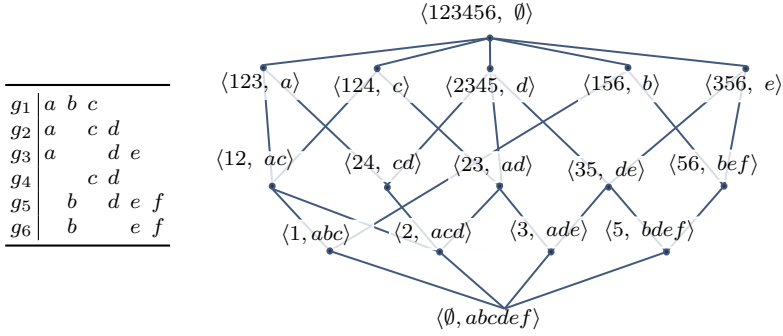}
	\end{minipage}
	\caption{A dataset of transactions (left) and its concept lattice (right)}
	\label{fig:example}
\end{figure}

\begin{table}[h]
	\caption{Non-trivial equivalence classes of closed itemsets for the dataset from Fig.~\ref{fig:example}. Closed itemsets are given in column $B$}
	\label{tab:keys}
	\centering
    \begin{tabular}{l|l|l|l}
        \toprule
        \begin{tabular}[c]{@{}l@{}}$B$\end{tabular}& \multicolumn{1}{c}{$pKey(B)$} & \multicolumn{1}{c}{$Key(B)$} & \multicolumn{1}{c}{$Equiv(B)$} \\\midrule
        $ade$ & $ae$ & $ae$ & $ae$, $ade$ \\\hline
        $abc$ & $ab$, $bc$ & $ab$, $bc$ & $ab$, $bc$, $abc$ \\\hline
        $bef$ & $f$ & $f$, $be$ & $f$, $be$, $ef$, $bf$, $bef$ \\\hline
        $bdef$ & $bd$, $df$ & $bd$, $df$ & $bd$, $df$, $def$, $bdf$, $bde$, $bdef$ \\\hline
        $abcdef$ & $af$, $ce$, $cf$&\begin{tabular}[c]{@{}l@{}}$af$, $ce$, $cf$,\\$abd$, $abe$, $bcd$\end{tabular}&  \begin{tabular}[c]{@{}l@{}}$af$, $ce$, $cf$, $abd$, $abe$, $abf$, $ace$, $acf$, $adf$,\\$aef$, $bcd$, $bce$, $bcf$, $bde$, $bdf$, $cde$, $cdf$,\\$cef$, $def$, $\ldots$, $abcdef^\ast$\end{tabular}  \\\bottomrule
    \end{tabular}\\
    $^\ast$ The equivalence class includes all itemsets that contain a key from $Key(abcdef)$.

\end{table}

Among closed itemsets with non-trivial equivalence classes, only $bef$ and $abcdef$ have passkey sets that differ from the key sets, i.e., $pKey(bef) = \{f\}$, $Key(bef) = \{f\} \cup \{ be\}$, $pKey(abcdef) = \{af, ce, cf\}$, $Key(abcdef) = pKey(abcdef) \cup \{abd, abe, bcd\}$. The keys from $Key(B)$ are minimal by set-inclusion, i.e., there is no element in the equivalence class that is contained in a key. For example, for key $be \in Key(bef)$ there is no itemset in $Equiv(bef)$, except $be$ itself, that is included in $be$.
The passkeys are keys that have the smallest size. Closed itemset $bef$ has one passkey $f$. Below (see Fig.9) we give an example of a case with unique passkey and exponentially many keys.

Thus, a passkey is the shortest representative of the equivalence class, i.e., the smallest lossless description (meaning that the whole lattice can be rebuild using the set of passkeys). Let us consider how passkeys can be used to define the closure structure.

\section{Closure structure and its properties}\label{sec:closure_structure}

\subsection{Keys and passkeys}

When closed itemsets are ordered by set inclusion, the structure containing the whole set of closed itemsets is a complete lattice called the concept lattice (see example in Figure~\ref{fig:example}). However, in a lattice one is usually interested in local positions of the closed itemsets -or concepts- most often in direct neighbors~\citep{buzmakov2014scalable}. Usually, the global position of the concepts w.r.t. the largest and smallest elements is not considered as useful information. Moreover, for a closed itemset, there can be exponentially many -- w.r.t. its size and support -- different ways to reach it from the top or the bottom of the lattice.

In this section, we introduce the ``closure structure'', that was firstly introduced in our previous work~\citep{mkn2020}, where we call it ``minimum closure structure''. Unlike the concept lattice, the closure structure is composed of levels that reflect the ``intrinsic'' complexity of itemsets. The complexity of a closed itemset $X$ is defined by the size of its passkeys, i.e., cardinality minimum generators of $X$. 

In the previous sections we consider different elements of the equivalence class, namely, generators, minimal generators (keys) and minimum generators (passkeys). We can organize the itemset space into levels based on minimal generators. In Fig.~\ref{fig:key_based_splitting} we show such a splitting. A closed itemset $B$ is located at level $k$ if there exists a key of $B$ of size $k$. The keys of closed itemsets that differ from the corresponding closed itemset are given in parentheses. Some of the itemsets have more than one key, e.g., $abc$ has keys $ab$ and $bc$, $bdef$ has keys $bd$ and $df$. The problem is that a closed itemset may be located in several levels. For example, $bef$ appears at the 1st and 2nd levels since it has keys $f$ and $be$, closed itemset $abcdef$ belongs to the 2nd and 3rd levels. 
\begin{figure}[t]
    \centering
\setlength{\tabcolsep}{4pt}
\setlength\extrarowheight{5pt}
\begin{tabular}{lcccccccccc}
$\mathcal{C}_1$&$a$&$b$&$c$&$d$&$e$&{\begin{tabular}[c]{@{}c@{}}$bef$\\$\textcolor{gray}{(f)}$\end{tabular}}&& & \\ \hline% \\ \hline \\
$\mathcal{C}_2$& \begin{tabular}[c]{@{}c@{}}$abc$\\$\textcolor{gray}{(ab,bc)}$\end{tabular}& \begin{tabular}[c]{@{}c@{}}$bdef$\\$\textcolor{gray}{(bd,df)}$\end{tabular} & \multicolumn{2}{c}{\begin{tabular}[c]{@{}c@{}}$abcdef$\\$\textcolor{gray}{(af,cf,ce)}$\end{tabular}}&\begin{tabular}[c]{@{}c@{}}$ade$\\$\textcolor{gray}{(ae)}$\end{tabular}& \begin{tabular}[c]{@{}c@{}}$bef$\\$\textcolor{gray}{(be)}$\end{tabular} &$ac$&$de$&$ad$&$cd$\\ \hline%\\ \hline \\
$\mathcal{C}_3$&$acd$&\multicolumn{5}{c}{\begin{tabular}[c]{@{}c@{}}$abcdef$\\$\textcolor{gray}{(abd, abe, bcd)}$\end{tabular}}
\end{tabular}

    \caption{Splitting a set of closed itemsets from Fig.~\ref{fig:example} by levels w.r.t. the size of their keys. The keys that differ from their closures are given in the parentheses. The set-inclusion relations between itemsets are not shown}
    \label{fig:key_based_splitting}
\end{figure}

To eliminate this kind of redundancy, we proposed to order the set of closed itemsets w.r.t. passkeys or ``cardinality minimum generators''.
Since all the passkeys of a closed itemset have the same size, the closed itemset belongs only to one level and hence the corresponding organization is a partition. 

Let us consider this structure more formally. 

%---------------------------------------------------- 
\subsection{Closure structure}\label{ssec:closure_structure}
%----------------------------------------------------

As it was mentioned above, all passkeys of a closed itemset have the same size. A passkey is a minimal generator of minimum cardinality. A passkey corresponds to the shortest itemset among all itemsets composing an equivalence class and characterizes the complexity of a closed itemset in terms of its size.

Let $\mathcal{C}$ be the set of all closed itemsets for the context $(G,M,I)$, and $pKey(B)$ be the passkey set of a closed itemset $B \in \mathcal{C}$. The function $Level: \mathcal{C} \rightarrow \{0, \ldots, |M|\}$ maps a closed itemset $B$ to the size of its passkeys, i.e. $Level(B) = |X|$, where $X \in pKey(B)$ is an arbitrary passkey from $pKey(B)$. The closed itemsets having passkeys of size $k$ constitute the closure level $k$.

\begin{definition}
Let $\mathcal{C}$ be a set of all closed itemsets for a context $(G, M, I)$. Then the \emph{closure level} $\mathcal{C}_k$ is the subset of all closed itemsets with passkeys of size $k$, i.e., $\mathcal{C}_k = \{B \mid B \in \mathcal{C}, \exists D \in pKey(B), |D| = k\}$. 
\end{definition}

In the definition above, in order to define the closure level of a closed itemset $B$, we use an arbitrary passkey $D \in pKey(B)$.
We distinguish passkeys only when we compute the closure structure and we discuss these details later in Section~\ref{sec:algorithm}.

\begin{proposition} 
The set of closure levels of the context $(G,M,I)$ defines a partition over the set of closed itemsets $\mathcal{C}$ of $(G,M,I)$ into sets of itemsets of increasing complexity w.r.t. the size of the passkeys. 
\label{prop:closure_structure}
\end{proposition} 

\begin{proof}
It follows from the definition of passkey that the passkeys associated with a closed itemset are all of the same size. Thus, each closed itemset belongs to only one closure level and the closure levels make a partition.
\smartqed
\end{proof}

\begin{definition}
The partition based on the set of closure levels of the context $(G,M,I)$ is called the \emph{closure structure} of $(G,M,I)$.

\end{definition}

The last non-empty level of the closure structure contains the largest passkeys. The size of such keys is the upper bound on the size of the generators of all closed itemsets in a given dataset, i.e., the least upper bound on the sizes of generators of the most ``complex'' closed itemsets.  

\begin{definition}
The \emph{closure index} of $\mathcal{C}$, denoted by $CI$, is the maximal number of non-empty level of the closure structure, i.e., $CI = \max \{ k \mid k = 1, \ldots, |M|, \, \mathcal{C}_k\neq \emptyset\}$. 
\label{def:CI}
\end{definition}

The closure index $CI$ is the largest size of an itemsets ``sufficient'' to represent any closed itemset. More formally, $\forall B \in \mathcal{C}, \; \exists X \in \mathcal{P(M)}, |X| \leq CI $ such that $X'' = B$. In such a way, the closure index characterizes complexity of the whole set of closed itemsets $\mathcal{C}$.

\noindent\textbf{Example.}
Let us consider the closure structure in Fig.~\ref{fig:closure_structure}.
Here, in contrast to the situation in Fig.~\ref{fig:key_based_splitting}, each closed itemset appears only once. A closed itemset from the $k$-th level may still have several passkeys of size $k$, but all these keys are of the same size. For example, $abc$ has two passkeys: $ab$ and $bc$, closed itemset $abcdef$ has three passkeys, i.e., $af$, $cf$ and $ce$.

\begin{figure}[t]
    \centering
\setlength{\tabcolsep}{4pt}
\setlength\extrarowheight{5pt}
\begin{tabular}{lcccccccccc}
$\mathcal{C}_1$&$a$&$b$&$c$&$d$&$e$&{\begin{tabular}[c]{@{}c@{}}$bef$\\$\textcolor{gray}{(f)}$\end{tabular}}& & &\\ \hline% \\ \hline \\
$\mathcal{C}_2$& \begin{tabular}[c]{@{}c@{}}$abc$\\$\textcolor{gray}{(ab,bc)}$\end{tabular}& \begin{tabular}[c]{@{}c@{}}$bdef$\\$\textcolor{gray}{(bd,df)}$\end{tabular} & \multicolumn{2}{c}{\begin{tabular}[c]{@{}c@{}}$abcdef$\\$\textcolor{gray}{(af,cf,ce)}$\end{tabular}}& \begin{tabular}[c]{@{}c@{}}$ade$\\$\textcolor{gray}{(ae)}$\end{tabular} &$ac$&$de$&$ad$&$cd$\\ \hline%\\ \hline \\
$\mathcal{C}_3$&$acd$&\\
\end{tabular}
   
    \caption{The closure structure for the dataset from Fig.~\ref{fig:example}. The passkeys that differ from their closures are highlighted in gray. The set-inclusion relations between itemsets are not shown}
    \label{fig:closure_structure}
\end{figure}

\subsection{Passkeys make an order ideal} \label{ssec:ideal}
In~\citep{stumme2002computing} it was shown that keys are downward closed, i.e., every proper subset of a key is a key. The latter means that given a certain key $X$ we may reconstruct $2^{|X|}-1$ other keys without computing any concept. In order theory the structure possessing this property is called an order ideal. 

\begin{definition}
Let $(\mathcal{P}, \leq)$ be a poset. A subset $\mathcal{I} \subseteq \mathcal{P}$ is called an \emph{order ideal} if $x \in \mathcal{I}$, $y \leq x $ then $y\in \mathcal{I}$.
\end{definition}

So, the set of all keys makes an order ideal in the set of all itemsets ordered by inclusion. The same holds for the set of all passkeys. 

\begin{proposition}
Let $\mathcal{I}_{P} = \bigcup_{k = 1, \ldots, |M|}\{D \in pKey(B) \mid B \in \mathcal{C}, |D| = k\}$ be the set of passkeys for the context $(G,M,I)$. Then $\mathcal{I}_{P}$ is an order ideal w.r.t. $\subseteq$. \label{prop:key_ideal}
\end{proposition}

%Let $\mathcal{I} = \bigcup_{k = 1, \ldots, |M|}\{D \in Key(B) \mid B \in \mathcal{C}, |D| = k\}$ be the set of keys and $\mathcal{I}_{P} = \bigcup_{k = 1, \ldots, |M|}\{D \in pKey(B) \mid B \in \mathcal{C}, |D| = k\}$ be the set of passkeys for the context $(G,M,I)$.
%Both $\mathcal{I}$ and $\mathcal{I}_{P}$ are order ideals w.r.t. $\subseteq$. 

\begin{proof}
Consider a passkey $X$ and a subset $Y\subseteq X$. Suppose that $Y$ is not a passkey, then there exists a passkey $Z: |Z| < |Y|,$ $Z'' = Y''$. Hence, the set $Z \cup X\setminus Y$ is a key of $X''$,  which is smaller in size than $Y\cup X\setminus Y = X$. This contradicts the fact that $X$ is a passkey.
\end{proof}

This property is important since, given a passkey, we may conclude that all its subsets are passkeys. 
For example, knowing that $acd$ is a passkey (see Fig.~\ref{fig:closure_structure}), we may conclude that all its proper subsets, namely $ac$, $ad$, $cd$, $a$, $c$, $d$ are also passkeys.
 
\subsection{Assessing data complexity} \label{ssec:lower_bound_lattice_size}

The closure index $CI$ of a formal context allows us to estimate the number of elements in the corresponding concept lattice. The estimates are based on evaluating ``the most massive'' part of a concept lattice, which is usually related to a Boolean sublattice.

\begin{definition}
A contranominal-scale context is a context of the form $(M, M, \neq)$, i.e., the context where 
a pair $x,y\in M$ belongs to the relation $(x,y)\in \neq$ iff $x\neq y$. The matrix of $(M, M, \neq)$ is filled with crosses except for the main diagonal.
\end{definition}

A contranominal-scale context represents an ``extreme case'' of context, where the concept lattice has $2^{|S|}$ elements, i.e., the maximal size, and each equivalence class is trivial, i.e., composed of the closed itemset itself. An example is given in Fig.~\ref{fig:contranominal}.

\begin{figure}
    \centering
    \begin{minipage}[t]{.4\textwidth}
    \centering
        \vspace{-8em}
        \begin{tabular}{l|llll}\toprule
         &$m_1$&$m_2$&$m_3$&$m_4$\\\midrule
        $g_1$&  &\X&\X&\X\\
        $g_2$&\X&  &\X&\X\\
        $g_3$&\X&\X&  &\X\\
        $g_4$&\X&\X&\X& \\\bottomrule
        \end{tabular}
       
    \end{minipage}
	\begin{minipage}[t]{.5\textwidth}
	    \centering
%	    \vspace{-3em}
	    \includegraphics[width = 0.8\textwidth]{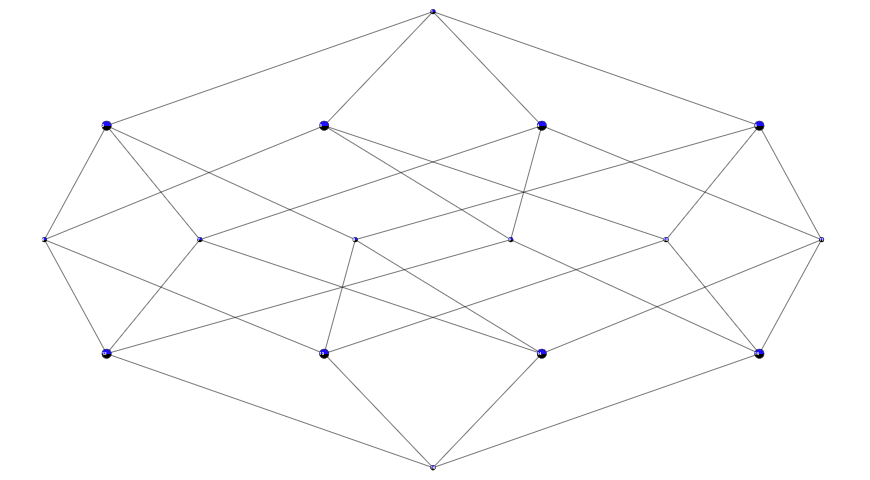}
	\end{minipage}
    \caption{A contranominal-scale context of size 4 and the corresponding concept lattice}
    \label{fig:contranominal}
\end{figure}

\begin{proposition}[lower bound on the lattice size]
The closure index $CI$ of $\mathbb{K} = (G,M,I)$ is a lower bound on the dimension of the inclusion-maximal Boolean sublattice of the concept lattice of $(G,M,I)$, so the tight lower bound on the number of closed itemsets is $2^{CI}$.
\end{proposition}

\begin{proof}
In~\citep{albano2017concept} (Lemma 6) it was proven that for a context $\mathbb{K} = (G, M,I)$ and $A \subseteq G$ there exists a contranominal-scale subcontext of $\mathbb{K}$ having $A$ as its object set if and only if $A$ is a minimal generator (key).

Following the definition of $CI$, the maximal size of passkeys in $\mathbb{K}$ is equal to $CI$, thus there is no passkey of larger size. Since the maximal size of a key is at least $CI$, the dimension of the maximal Boolean sublattice is at least $CI$.

Hence, the number of all closed itemsets $\mathcal{C}$ of $(G,M,I)$ is at least $2^{CI}$.
\smartqed 
\end{proof}

In many situations, a binarization is used for mining patterns in numerical data.
We consider the \textit{simplest binarization} where each interval of an attribute is mapped to a new binary attribute, as it was explained in Section~\ref{sec:basics}.
In this case the closure index $CI$ cannot be greater than the number of attributes in the original dataset. Below, we consider this extreme case, where $CI$ is maximal.

\begin{proposition}[complexity of binarized datasets]
Let $\mathbb{K}_{num} = (G,S,W, I_{num})$ be a many-valued context and $\mathbb{K} = (G,M,I)$ be its simplest binarization.
If the closure index $CI$ of $\mathbb{K}$ is equal to the number of attributes $|S|$ in $\mathbb{K}_{num}$, then there exist $CI$ objects in $G$ such that their descriptions form a contranominal-scale subcontext of $\mathbb{K}$ and any pair of these descriptions in $\mathbb{K}_{num}$ differs in two attributes.
\label{prop:ci_binarized}
\end{proposition}

\begin{proof}
Let us suppose that the number of original numerical attributes is equal to the closure index $CI$ of $\mathbb{K}$.
Thus, there exists a contranominal-scale context of size $CI$, where any pair of the corresponding objects have $CI-1$ matching binary values. Because of the bijection between the intervals of discretization and binary attributes, at least $CI$ object descriptions differ in a specific pair of values.
\smartqed 
\end{proof}

\noindent\textbf{Example}. Table~\ref{tab:many_valued_context} (right) represents a subcontext where any two object descriptions differ in two attribute values (up to intervals used for binarization). 

\subsection{Closeness under sampling} 

In many practical situations, datasets are large and a way to efficiently analyze such datasets is to use data sampling, i.e., building a model based on data fragments rather than on the whole dataset.
There is a lot of approaches to sample data, depending on the task, data specificity and the selected learning models.
%However, the problem of counting the number of closed itemsets for a given dataset is known to be as hard as $\#RH\Pi \subset \#P$~\citep{dyer2004relative} and there is no a fully polynomial randomized approximation scheme to solve this problem~\citep{boley2010formal}. 

In this section we consider a particular type of sampling that consists in selecting a certain amount of objects keeping the complete attribute space, i.e., we do not sample the attributes.

Let $\mathbb{K} = (G,M, I)$ be a formal context. Then an \textit{object sampling} in the context $\mathbb{K}$ is a context $\mathbb{K}_s = (G_s, M, I_s)$, where $G_s\subseteq G$, and $I_s = \{ g I m \mid g \in G_s\}$.

\begin{proposition} 
Itemsets closed in $\mathbb{K}_s$ are closed in $\mathbb{K}$.
\end{proposition}

\begin{proof}
Let $\mathbb{K} = (G,M,I)$ be a formal context, $\mathbb{K}_s = (G_s, M, I_s)$ be a formal subcontext, $(A_s,B_s)$ be a formal concept of $\mathbb{K}_s$.
We denote the derivation operators $(\cdot)'$ in context $\mathbb{K}$ by indexing them with the context name, i.e., $(\cdot)'_\mathbb{K}$.

Suppose that $B_s$ is not closed in $\mathbb{K}$, i.e., there exists a concept $(A,B)$ in $\mathbb{K}$ such that $B_s \subset B$ and $B_s \in Equiv(B)$ in $\mathbb{K}$. The later means that they have the same extent in $\mathbb{K}$, i.e., $(B)'_{\mathbb{K}} = (B_s)'_{\mathbb{K}}$. Since $B_s$ is closed in  $\mathbb{K}_s$ it follows directly that $B\not\in Equiv(B_s)$, and from $B_s \subset B$ it follows that $(B)'_{\mathbb{K}_s} \subset (B_s)'_{\mathbb{K}_s}$. Hence, $\exists g \in (B_s)'_{\mathbb{K}_s} \subseteq (B_s)'_{\mathbb{K}}$ such that $g \not\in (B)'_{\mathbb{K}}$. This contradicts the fact that $B_s \in Equiv(B)$, thus our assumption is wrong and $B_s$ is closed in $\mathbb{K}$.
\smartqed 
\end{proof}

\begin{proposition} 
For any closed itemset $B$, $Level(B)_{\mathbb{K}_s} \leq Level(B)_{\mathbb{K}}$, i.e., the passkey size of $B$ for a sampled context is not greater than the passkey size of $B$ for the whole context. 
\end{proposition}

\begin{proof}

Suppose the opposite, i.e., that $Level(B)_{\mathbb{K}_s} > Level(B)_{\mathbb{K}}$. Thus, there exists a passkey $X \in pKey(B)_{\mathbb{K}_s}$ and a passkey $Y \in pKey(B)_{\mathbb{K}}$ such that $|X| > |Y|$. The latter means that $(Y)''_{\mathbb{K}} = B$. However, from Proposition~\ref{prop:key_ideal} and the definition of the passkey, for any $Y \subset B$, such that $ |X| > |Y|$, $Y \not\in Equiv(B)_{\mathbb{K}_s}$, thus, $(B)'_{\mathbb{K}_s} \subset (Y)'_{\mathbb{K}_s}$. It means that in the augmented context $\mathbb{K}$, $(B)'_{\mathbb{K}} \not\supseteq (Y)'_{\mathbb{K}}$, i.e., the extent of $Y$ cannot be equal or included into extent of $B$, thus, $Y \not\in Equiv(B)_{\mathbb{K}}$. Thus, there is no a passkey $Y \in pKey(B)_{\mathbb{K}}$, such that $|Y| < |X|$. Hence, our assumption is wrong and $Level(B)_{\mathbb{K}_s} \leq Level(B)_{\mathbb{K}}$.
\smartqed 
\end{proof}

To illustrate that the closure level $Level(B)_{\mathbb{K}_s}$ (see Section~\ref{ssec:closure_structure}) of a concept $(A,B)$ in a sampled context $\mathbb{K}_s$ is a lower bound on the closure level $Level(B)_{\mathbb{K}}$ of the concept in the larger context $\mathbb{K}$, we consider an example from Fig.~\ref{fig:example}. The closed itemset $acd$ belongs to the 3rd level. Fig.~\ref{fig:sampled_structures} shows the closure structures for sampled context with objects $G_s = \{g_1, g_2, g_3, g_4\}$ (left) and $G_s = \{g_1, g_2, g_3, g_5\}$ (right). Closed itemset $acd$ belongs to the 3rd and 2nd levels, respectively. Hence, a sampled context may be used to compute closed itemsets, for any closed itemset $B$, $Level(B)_{\mathbb{K}_s} \leq Level(B)_{\mathbb{K}}$.

\begin{figure}
    \centering
    \begin{minipage}[t]{.45\textwidth}
        \centering
\setlength{\tabcolsep}{4pt}
\setlength\extrarowheight{5pt}
\begin{tabular}{lcccccc}
    $\mathcal{C}_1$&$a$&\begin{tabular}[c]{@{}c@{}}$abc$\\$\textcolor{gray}{(b)}$\end{tabular}&$c$&$d$&\begin{tabular}[c]{@{}c@{}}$ade$\\$\textcolor{gray}{(e)}$\end{tabular}&$f$\\ \hline
    $\mathcal{C}_2$& \multicolumn{2}{c}{$ac$} & \multicolumn{2}{c}{$ad$} & \multicolumn{2}{c}{$cd$} \\\hline 
    $\mathcal{C}_3$& \multicolumn{6}{c}{$acd$}
\end{tabular} 
    \end{minipage}
    \begin{minipage}[t]{.45\textwidth}
        \centering
\setlength{\tabcolsep}{4pt}
\setlength\extrarowheight{5pt}
\begin{tabular}{lcccccc}
    $\mathcal{C}_1$&$a$&$b$&\begin{tabular}[c]{@{}c@{}}$ac$\\\textcolor{gray}{$(c)$}\end{tabular}&$d$&\begin{tabular}[c]{@{}c@{}}$de$\\\textcolor{gray}{$(e)$}\end{tabular}&\begin{tabular}[c]{@{}c@{}}bdef\\$\textcolor{gray}{(f)}$\end{tabular}\\ \hline 
    $\mathcal{C}_2$&$ad$&\begin{tabular}[c]{@{}c@{}}$ade$\\$\textcolor{gray}{ae}$\end{tabular}&\begin{tabular}[c]{@{}c@{}}$acd$\\$\textcolor{gray}{cd}$\end{tabular}&\begin{tabular}[c]{@{}c@{}}$abc$\\$\textcolor{gray}{(ab,bc)}$\end{tabular} & \multicolumn{2}{c}{\begin{tabular}[c]{@{}c@{}}$abcdef$ \\$\textcolor{gray}{(af,ce,cf,bf)}$\end{tabular}}
\end{tabular} 
    \end{minipage}
    \caption{The closure structures computed on sampled contexts $\mathbb{K} = (G_s, M, I)$ with the sets  $G_s = \{g_1, g_2, g_3, g_4\}$ (left) and $G_s = \{g_1, g_2, g_3, g_5\}$ (right). The passkeys that differ from their closure are highlighted in gray}
    \label{fig:sampled_structures}
\end{figure}

\section{The GDPM algorithm}\label{sec:algorithm}

In this section we introduce the GDPM algorithm --for Gradual Discovery in Pattern Mining-- which computes the closure structure.
We analyze the complexity of GDPM and then consider computational issues related to key-based structures.
Finally we discuss how the closure structure may reduce the computational complexity of a mining problem.

\subsection{Computing the closure structure with GDPM}\label{ssec:gdpm}

The GDPM algorithm enumerates closed itemsets w.r.t. the closure levels based on lexicographically smallest passkeys. At each iteration GDPM computes the next closure level $\mathcal{C}_k$ using the itemsets from the current closure level $\mathcal{C}_{k-1}$. The algorithm can be stopped at any time. It returns the computed closure levels and, optionally, one passkey per closed itemset. GDPM requires that all attributes from $M$ be linearly ordered, i.e., $m_1 < m_2<\ldots < m_k$. Given this order, we say that an itemset $X_1 \subseteq M$ is \textit{lexicographically smaller} than an itemset $X_2\subseteq M$, if the smallest differing element belongs to $X_1$. For example, $abf$ is lexicographically smaller than $ac$ because $b<c$.

One of the main efficiency concerns in enumeration algorithms is to ensure the uniqueness of enumerated itemsets. To solve this issue, GDPM keeps all generated closed itemsets in a lexicographic tree to ensure that each closed itemset is generated only once. A \textit{lexicographic tree}~\citep{zaks1980lexicographic} --\textit{trie} or \textit{prefix tree}-- is a tree that has the following properties:  (i) the root corresponds to the null itemset; (ii) a node in the tree corresponds to an itemset, (iii) the parent of a node $i_1 \ldots i_{\ell-1} i_\ell$ is the node $i_1 \ldots i_{\ell-1}$.

%In the previous sections we use $\mathcal{C}$ to denote a set of closed itemsets $B$, in this section $\mathcal{C}$ denotes a set of formal concepts $(A,B)$, i.e., a closed itemset $B$ and its extent $A$.  

The pseudocode of GDPM is given in Algorithm~\ref{alg:gdpm}.
GDPM is a breadth-first search algorithm, which requires all attributes from $M$ to be linearly ordered.
The closure structure is computed sequentially level by level, level $\mathcal{C}_k$ is computed by adding one by one attributes to each closed itemset from $\mathcal{C}_{k-1}$.
A trie $\mathcal{T}_{k}$ is used to perform the duplicate test.
In the beginning, a trie consists of the root labeled by the empty itemset.
Then, at each iteration (line~\ref{alg:gdpm_main_begin}-\ref{alg:gdpm_main_end}) the trie grows by including the closed itemsets from the closure level $\mathcal{C}_k$ that is currently being computed.
The algorithm can be halted at any level.

\begin{algorithm}
    \caption{GDPM (main loop)}
    \label{alg:gdpm}
    %\SetAlgoLined
    %\SetKwInOut{Input}{Input}\SetKwInOut{Output}{Output}
    \begin{algorithmic}[1]
    \REQUIRE{$\mathbb{K} = (G,M,I)$, a formal context, $k_{max}$ the maximal closure level to compute}
    \ENSURE{$\mathcal{C}_{k}$, $k = 0,\ldots, k_{max}$, closed itemsets of the levels up to $k_{max}$}
        \STATE $k \leftarrow 0$ \;
        \STATE $\mathcal{C}_k \leftarrow \{ \emptyset\}$ \;
		\STATE $\mathcal{T}_{k} \leftarrow \emptyset$ \;
        \REPEAT{ \label{alg:gdpm_main_begin}
            \STATE $k \leftarrow k + 1$ \; 
            \STATE $\mathcal{C}_k\leftarrow GDPM\mbox{-}INT(\mathcal{C}_{k-1}, \mathcal{T}_{k - 1}, \mathbb{K})$ \hspace{.2em}  // Algorithm~\ref{alg:gdpm_int}
        }\UNTIL{$|\mathcal{C}_k|>0$ and $k < k_{max}$}\label{alg:gdpm_main_end}
		\RETURN $\mathcal{C}_{k}$
	\end{algorithmic}
\end{algorithm}

The function GDPM-INT for computing level $\mathcal{C}_k$ using itemsets from level $\mathcal{C}_{k-1}$ is described in Algorithm~\ref{alg:gdpm_int}. The candidates for the closure level $\mathcal{C}_k$ are generated using a closed itemset $B$ of the current closure level $\mathcal{C}_{k-1}$ and item $m$ that is not included into $B$. For a candidate $B \cup \{m\}$ GDPM computes its closure (line~\ref{line:closure}), and then checks if this closed itemset is included in the trie (line~\ref{line:tree_check}). The closed itemsets that are not in the trie are added both to the trie $\mathcal{T}_k$ and the closure level $\mathcal{C}_k$ (line~\ref{line:add_to_tree} and~\ref{line:add_to_cl}, respectively).

\begin{algorithm}[H]
    \caption{GDPM-INT\hspace{3.em} // GDPM-EXT}
    \label{alg:gdpm_int}
    \begin{algorithmic}[1]
        \REQUIRE $\mathcal{C}_{k-1}$, closed itemsets of the level $k-1$,\\
        $\mathcal{T}_{k-1}$, the trie containing all closed itemsets $\bigcup_{i = 1, \ldots, k-1}\{B \mid B \in \mathcal{C}_i\}$ 
        \ENSURE $\mathcal{C}_{k}$, closed itemsets of the level $k$
        \STATE  $\mathcal{C}_{k} \leftarrow \emptyset $
        \STATE  $\mathcal{T}_{k} \leftarrow \mathcal{T}_{k-1}$
        \FORALL{$B \in \mathcal{C}_{k-1}$}
            \FORALL{$m \in M \setminus B$}\label{line:candidates} 
                \STATE $B_c \leftarrow (B \cup \{m\})''$\label{line:closure}  \hspace{4.2em}  // $A_c \leftarrow (B \cup \{m\})'$
                    \IF{$B_c \notin \mathcal{T}_{k} \textbf{ then  }  \hspace{5em} //\textbf{  if }A_c \notin \mathcal{T}_{k}$}\label{line:tree_check} 
                        \STATE $add(\mathcal{T}_{k},B_c)$\label{line:add_to_tree}  \hspace{5.5em}   // \hspace{1.5em} $add(\mathcal{T}_{k},A_c)$
                        \STATE $add(\mathcal{C}_{k},B_c)$\label{line:add_to_cl}  \hspace{5.4em}   // \hspace{1.5em}  $add(\mathcal{C}_{k},A_c')$
                    \ENDIF	
            \ENDFOR
        \ENDFOR
        \RETURN $\mathcal{C}_{k}$  
    \end{algorithmic}
    \end{algorithm}

\begin{figure}
    \centering
    \input{gdpm_int_steps}
    \caption{An iteration in the GDPM-INT algorithm for computing itemsets of the 2nd level using closed itemset $b$. The duplicate test (line 6) is performed after computing both the extent and intent (line 5). Elements which are newly added to $\mathcal{T}_2$ are highlighted in red}
    \label{fig:gdpm_me_execution}

    \centering
    \input{gdpm_ext_steps}
    \caption{An iteration of the GDPM-EXT algorithm for computing itemsets of the 2nd level using closed itemset $b$. The duplicate test (line 6) is performed after computing an extent (line 5). The intent is computed only if the extent has not been yet created (line 8). An altered element in $\mathcal{T}_2$ is highlighted in red}
    \label{fig:gdpm_f_execution}
\end{figure}

\noindent\textbf{Example.}
Let us consider some intermediate steps of GDPM.
We use the dataset from Fig.~\ref{fig:example} and consider steps when a closed itemset $b \in \mathcal{C}_1$ is used to compute closed itemsets from $\mathcal{C}_2$.
The result of the 1st call of GDPM-INT is $\mathcal{C}_1 = \{a,b,c,d,e,bef\}$.
In the 2nd call, after considering the candidates $\{a, m\}$ where $m \in \{b,c,d,e,f\}$, five new itemsets are added to the trie, namely $abc$, $abcdef$, $ac$, $ad$, $ade$.
Then, the algorithm proceeds to compute the closed itemsets based on candidates $\{b,m\}$, $m \in \{c,d,e,f\}$.
The execution tree of GDPM-INT and the trie are given in Fig.~\ref{fig:gdpm_me_execution}.
We can see that for an itemset $\{b,m\}$, $m \in \{c,d,e,f\}$ GDPM-INT computes the closure $\{b,m\}''$ in line 5, and only after that checks the existence of $\{b,m\}''$ in the trie $\mathcal{T}_2$.
The duplicate test $B_c\not\in \mathcal{T}_2$ is performed in $O(|M|)$ steps.

Using the execution tree, we can also reconstruct the lexicographically smallest passkey. For example, itemset $bdef$ was generated in the branches labeled by ``$b$'' (to get a closed itemset $b$) and then ``$d$''. Thus, the lexicographically smallest passkey of $bdef$ is $bd$.

%The proposed algorithm is a variation of Close-By-One (CbO)~\citep{kuznetsov1993fast}. The call trees (records of execution) of both algorithms are the same. It means that a closed itemset is computed based on the same passkey (provided that the attributes in $I$ are considered in the same order). They differ in the test for the duplicate of closed itemsets, namely the canonicity test (in CbO) and lexicographic tree (in the proposed version). 

\subsection{The extent-based version of GDPM}\label{ssec:gdpm_ext}

In the previous section we considered the GDPM algorithm and the sub-algorithm GDPM-INT for computing a single closure level. 
To perform the duplicate test, GDPM-INT uses a trie based on closed itemsets or intents.
Below we introduce an alternative version ``GDPM-EXT'', where, instead of intents, the trie contains the extents of closed itemsets.

The difference is that instead of computing the closure $B_c = (B \cup \{m\})''$, GDPM-EXT  computes only the extent $A_c = (B \cup \{m\})'$ (line~\ref{line:closure}). Then GDPM checks if $A_c$ exists in the trie $\mathcal{T}_k$ (line~\ref{line:tree_check}) and inserts $A_c$ into $\mathcal{T}_k$ if it is not found (line~\ref{line:add_to_tree}). For the inserted extent $A_c$, the corresponding closed itemset $A_c'$ is added to the closure level $\mathcal{C}_k$ (line~\ref{line:add_to_cl}).

\noindent\textbf{Example.} Let us consider how GDPM-EXT works by means of a small example. We consider same steps as for GDPM-INT, i.e., when a closed itemset $b \in \mathcal{C}_1$ is used to compute closed itemsets from $\mathcal{C}_2$. The execution tree of GDPM-EXT and the trie are given in Fig.~\ref{fig:gdpm_f_execution}. Both GDPM-INT and GDPM-EXT generate closed itemsets in the same order, i.e., the result of the 1st call for both of them is $\mathcal{C}_1 = \{a,b,c,d,e,bef\}$. In the 2nd call after considering candidates $\{a, m\}$, $m \in \{b,c,d,e,f\}$ the following new extents are added to the trie: $1$, $\emptyset$, $12$, $23$, $3$. The corresponding closed itemsets, added to $\mathcal{C}_k$, are the same as for GDPM-INT, i.e., $abc$, $abcdef$, $ac$, $ad$, $ade$. Then, GDPM-EXT proceeds to compute the closed itemsets based on $\{b,m\}$, $m \in \{c,d,e,f\}$. The execution tree in Fig.~\ref{fig:gdpm_f_execution} shows that for an itemset $\{b,m\}$, $m \in \{c,d,e,f\}$ GDPM-EXT computes only $\{b,m\}'$ and then checks its existence in the trie $\mathcal{T}_2$. Comparing the execution trees in Fig.~\ref{fig:gdpm_me_execution} and~\ref{fig:gdpm_f_execution}, we can see that the branches  ``$c$'', ``$e$'', ``$f$'' of GDMP-EXT are shorter than the branches of GDPM-INT, thus GDPM-EXT may work faster than GDMP-INT. However, the duplicate test $X\in \mathcal{T}_2$ is performed in $O(|M|)$ and $O(|G|)$ steps for GDPM-INT and GDPM-EXT, respectively. Thus, in case where $|M|<< |G|$ GDPM-INT may require smaller space than GDPM-EXT. Let us consider time and space complexity of GDMP in detail.

\subsection{Complexity of GDPM}\label{ssec:complexity_gdpm}
\paragraph{The estimates of the size of the next output.}
Given the input $\mathcal{C}_k$ (which is the output of the previous iteration), Algorithm~\ref{alg:gdpm} generates $|\mathcal{C}_{k + 1}|$ concepts at the next iteration, the maximal number of the concepts is  $\sum_{(A,B) \in \mathcal{C}_{k}} |M \setminus B|$. More generally, the number of concepts (or closed itemsets) at level $k + n$, $n \in \mathbb{N}$, is $O( |\mathcal{C}_{k}||M|^{n})$.

\paragraph{The worst-case space complexity of GDPM.} At each iteration Algorithm~\ref{alg:gdpm} stores the output of the previous iteration $\mathcal{C}_k$, the current output $\mathcal{C}_{k + 1}$, and the trie with the closed itemsets/extents for GDPM-INT/EXT, respectively. Thus, at the $(k+1)$-th iteration GDPM requires $O\left((|\mathcal{C}_{k}| + |\mathcal{C}_{k+1}|) \cdot (|G| + |M|)\right)$ space to store the closed itemsets and $O(\sum_{i = 1}^{k + 1}|\mathcal{C}_i|\cdot |M|)$ or $O(\sum_{i = 1}^{k + 1}|\mathcal{C}_i|\cdot |G|)$ space to store the trie for GDPM-INT and GDPM-EXT, respectively. If the size of closed itemsets is smaller than the sizes of their extents, GDPM-INT may require less space than GDPM-EXT.

\paragraph{The worst-case time complexity of GDPM.}
To compute concepts $\mathcal{C}_{k+1}$ using $\mathcal{C}_k$ as output, one needs to check $O(|\mathcal{C}_{k}|\cdot|M|)$ candidates. For each candidate GDPM-INT computes both the extent and the intent, which takes $O(|G||M|)$. Then, the duplicate test takes $O(|M|)$. GDPM-EXT computes for each candidate only the extent, that takes $O(|G|)$ time, and then checks for duplicates in $O(|G|)$. Only for the candidates that pass the duplicate test it computes the intent taking $O(|G||M|)$ time. Despite the fact that dublicates are identified faster by GDPM-EXT, asymptotically computing a closed itemset by both algorithms takes $O(|G||M|)$.
Hence, the total worst-case complexity of the $k+1$-th iteration of GDPM-INT and GDPM-EXT is $O(|\mathcal{C}_k|\cdot |G|\cdot  |M|^2)$. 

\subsection{Related approaches to key-based enumeration}\label{ssec:computing_cs_keys_min_keys}

In this section we consider related approaches to enumerating closed itemsets based on keys and the difference between them and GDPM.

As it was discussed in Section~\ref{sec:related_work}, one of the main efficiency concerns in enumeration algorithms is to ensure the uniqueness of enumerated closed itemsets (concepts). The canonicity test is known to be an efficient solution to this problem~\citep{kuznetsov2002comparing}. For performing the canonicity test the attributes from $M$ are required to be linearly ordered (as for GDPM).

The canonicity test consists in checking if a newly computed closure of a generator is lexicographically greater than the generator itself. The canonicity test is applicable in depth-first search algorithms where itemsets are extended at each iteration by adding an item, thus the closed itemsets are enumerated in the ``lexicographically smallest branch'' of the execution tree. Unfortunately, we cannot use this approach when we generate closed itemsets based on keys and passkeys. Let us consider an example from Fig.~\ref{fig:example} to illustrate this problem. The closed itemset $ade$ has one key (which is a passkey) $ae$. Using the passkey $ae$, we get the closure $ade$ that is lexicographically smaller than the passkey itself, thus, the canonicity test is not passed. Hence, no key and passkey can be used to generate $ade$ relying on the canonisity test.

The keys (including passkeys) are computed by a breadth-first search algorithm in order to ensure enumeration of the keys before other generators. To the best of our knowledge, the first algorithm for computing closed itemsets by enumerating the keys was proposed in~\citep{pasquier1999discovering} and called A-Close. An alternative principle of enumeration was proposed in the Pascal algorithm~\citep{bastide2000mining} and improved in the Titanic algorithm~\citep{stumme2002computing}. These three algorithms rely on the fact that keys make an order ideal, i.e., each subset of a key is a key (see Proposition~\ref{prop:key_ideal}). 

In~\citep{szathmary2009efficient} a depth-first search algorithm, called Talky-G, for enumerating keys was proposed. This algorithm traverses the search space in so-called ``reverse pre-order'' and checks if a subset is a key by means of the hash table containing all enumerated previously keys. 

The main challenge related to a \emph{correct} enumeration of keys is to avoid enumeration of duplicates and non-keys. We call an enumeration algorithm \emph{local} if it can eliminate a non-key itemset $X$ by considering its proper subsets $X \setminus \{m\}$, $m \in X$. The \emph{global} enumerating algorithm is defined dually. 

The main advantage of the local enumeration is lower space requirements, since one needs to store only a subset of  previously enumerated itemsets. Actually, to check if an itemset $X$ is a key, one needs to check whether all its subsets of size $|X|-1$ are keys.

The key can be enumerated locally using a breadth-first search approach. In this case, the keys of size $k$ are generated at the $k$-th iteration. To check if a candidate is a key one needs to store only the output of the previous iteration, i.e., the keys of size $k-1$. The depth-first approaches may work faster, but they require storage of all enumerated keys.

Let us consider the principles of key enumeration by a breadth-first algorithm called Titanic. The key enumeration is performed level by level and includes the following steps:

\begin{enumerate}
    \item Candidate generation. The set of candidates $\mathcal{S}$ is generated by merging all keys from the previous level $\mathcal{K}_{k-1}$ that differ by one item: $\mathcal{S}  =\{\{m_1 < \ldots < m_k \} \mid $
    				$  \{m_1, \ldots,  m_{k - 2},  m_{k- 1} \}, $ $\{m_1, \ldots,  m_{k - 2},  m_{k} \} \in \mathcal{K}_{k-1}\}$, where $\mathcal{K}_{k-1}$ is the set of all keys of size $k-1$.
    \item Verification of the ideal property, i.e., for each candidate one needs to check whether each of the $k$ proper subsets of size $(k-1)$ is included into $\mathcal{K}_{k-1}$.
    \item Support-closeness verification, i.e., comparison of the actual support with the minimal support of keys from $\mathcal{K}_{k-1}$ contained in the current candidate.
\end{enumerate}

The verification of keys (Step 2 and 3) is performed by considering $k$ immediate subsets. The size of the largest possible level $\mathcal{K}_{|M|/2}$ is $O\left(\frac{2^{|M|+1/2}}{\sqrt{|M|\pi}}\right)$ (this follows directly from the Stirling's formula). Thus, the worst-case time complexity of the first step (consisting in considering all pairs of keys from $\mathcal{K}_{k-1}$ and comparing their elements) is $O\left(\frac{2^{2|M| + 1}}{|M|\pi }|M|\right)$. However, this step can be optimized by storing all keys in a trie and merging each itemset corresponding to a key with its siblings that are lexicographically greater than the itemset. Completeness of this enumeration follows directly from the definition of the order ideal, i.e., any key of size $k$ contains two keys of size $k-1$ that differ only by the last elements. Thus, the worst-case time complexity of the optimized first step is $O\left(\frac{2^{|M|+1/2}}{\sqrt{|M|\pi}}|M|^2\right)$, because the multiplier $\frac{2^{|M|+1/2}}{\sqrt{|M|\pi}}$ is replaced by the maximal number of siblings $|M|$.
This solution is also related to the problem of a ``single parent'' discussed in~\citep{arimura2009polynomial}, where the authors proposed a depth-first search algorithm, called  CloGenDFS, that traverses a search subspace (called a family tree) in such a way that each closed itemset has only one parent and thus is generated only one time. In the case of closed itemsets, it ensures the absence of duplicates. However, in the case of keys, using a breadth-first search traversal we may only expect a better execution time because each $k$-sized itemset will be generated only once, but to eliminate non-key elements, one should perform Steps 2 and 3.

Despite the locality of the Titanic algorithm, the worst-case space complexity is not polynomial and proportional to the size of the largest level, i.e., $O\left(\frac{2^{|M|+1/2}}{\sqrt{\pi|M|}}\right)$.

Even if the splitting induced by the sizes of keys is very similar to the closure structure (e.g., compare Fig.~\ref{fig:key_based_splitting} and~\ref{fig:closure_structure}), there is a principal difference in the strategies for ensuring the correct enumeration of keys and passkeys. The elements of the closure structure cannot be enumerated locally, since ``minimality in size'' cannot be verified by checking only immediate subsets. 

\begin{figure}
    \centering
    \begin{minipage}{.25\textwidth}
        \begin{tabular}{l|llll}
            \toprule
            $g_1$&$a$&$b$&$c$&$d$\\
            $g_2$& &$b$&$c$&  \\
            $g_3$&$a$&  &$c$&  \\
            $g_4$&$a$&$b$&  & \\
            \bottomrule
        \end{tabular}
    \end{minipage}
    \begin{minipage}{.5\textwidth}
        \centering
\setlength{\tabcolsep}{4pt}
\setlength\extrarowheight{5pt}
\begin{tabular}{lcccc}
    $\mathcal{C}_1$&$a$&$b$&$c$& \begin{tabular}[c]{@{}c@{}}$abcd$\\$\textcolor{gray}{(d)}$\end{tabular} \\ \hline 
    $\mathcal{C}_2$&$ab$&$bc$&$ac$&  \\  \hline 
    $\mathcal{C}_3$& \multicolumn{3}{c}{\begin{tabular}[c]{@{}c@{}}$abcd$\\$\textcolor{gray}{(abc)}$\end{tabular}} & 
\end{tabular}
    \end{minipage}
    \caption{A dataset and the corresponding closure structure}
    \label{fig:local_global_enumeration}
\end{figure}

We explain the difference by a small example from Fig.~\ref{fig:local_global_enumeration}. To check, whether an itemset $X$ is a key one needs to verify that all $|X|$ immediate itemsets of size $|X| - 1$ are in this structure and belong to different equivalence classes than $X$. For example, when computing $\mathcal{C}_3$, to check whether $abc$ is a key, it is enough to ensure that $ab$, $ac$ and $bc$ are in $\mathcal{K}_2$ and belong to equivalence classes different from that of $abc$. For passkeys one also needs to check ``minimality in size'' in the equivalence class. Key $abc$ does not comply with the last constraint, since the smallest key is $d \in Equiv(abc'')$. To verify it one should check keys in $\mathcal{K}_1$. Thus, for computing the closure structure one needs to store all closed itemsets corresponding to enumerated passkeys, and the worst-case space complexity of the algorithm is proportional to the number of closed itemsets $O(2^{|M|})$. However, the related problem --enumeration of the passkeys-- may require much less space than enumeration of keys, since the number of keys can be exponentially larger than the number of passkeys (see Section~\ref{ssec:computing_min_keys}, Proposition~\ref{prop:one_vs_exp}).

%Thus, despite the fact that a passkey-based structure is not larger than the corresponding key-based structure, the worst-case space complexity of algorithms for enumerating a passkey-based structure is higher than that of enumerating the corresponding key-based structure.

%
\subsection{Computing passkeys. Towards polynomial complexity}\label{ssec:computing_min_keys}
%

% In the previous sections we presented the GDPM algorithm, which generates closed itemsets by levels w.r.t. the size of the passkeys, and, optionally, returns the lexicographically smallest passkey for each closed itemset.
The complexity of the GDPM algorithm is higher than that of simple enumerators of closed itemsets.
However, knowing the closure level to which a closed itemset belongs allows us to significantly reduce the complexity of enumeration of its passkeys.    

Keys play an important role in FCA and in the related fields~\citep{han2000mining}, e.g., graph theory, database theory, etc. Below we list the complexity of some problems related to computing or enumerating keys or passkeys.
It is known that computing keys and generators of a closed itemset is intractable.
In~\citep{kuznetsov2007stability}, it was shown that the problem of enumerating keys is \#P-complete\footnote{%
In particular, see Theorem 4, where the numerator of integral stability index is the number of generators of a closed itemset.
}.
The authors of~\citep{gunopulos2003discovering} showed that the problem of enumerating keys is \#P-complete.
More precisely, they studied the problem of enumerating minimal keys of a relational system.
Because of the correspondence between minimal keys of a relational system and keys --in the sense we use in this paper-- one can use this result to compute the complexity of enumerating keys.
Finally, in~\citep{hermann2008complexity} the authors studied the problem of enumerating passkeys and showed that this problem is \#coNP-complete.

Let us consider the complexity of the problem of computing the passkeys.

\begin{problem} PASSKEY\\
    INPUT: context $\mathbb{K} = (G,M,I)$ and subset of attributes $X\subseteq M$.\\
    QUESTION: Is $X$ a passkey, i.e., it is a key of $X''$ and there is no key of $X''$  of size less than $|X|$?
\label{prob:passkey}    
\end{problem}

\begin{proposition} The PASSKEY problem is coNP-complete.
\end{proposition}
 
\begin{proof} We introduce the following auxiliary problems:
 
\begin{enumerate}
    \item[]
    \begin{problem}KEY OF SIZE AT MOST $K$\\
        INPUT:  context $\mathbb{K} = (G,M,I)$, intent $B\subseteq M$ and natural number $k \in \mathbb{N}$.\\
        QUESTION: Is there a key of $B$ of size less than or equal to $k$?
        \label{prob:key_k}
    \end{problem}
    
    \item[]
    \begin{problem} NO KEY OF SIZE AT MOST $K$\\
        INPUT: context $\mathbb{K} = (G,M,I)$, intent $B\subseteq M$ and natural number $k \in \mathbb{N}$.\\
        QUESTION: Is it true that there is no key of $B$ of size less than or equal to $k$?\\
        \label{prob:no_key_k}
    \end{problem}
\end{enumerate}

Problem~\ref{prob:key_k} is trivially polynomially equivalent to the famous $NP$-complete SET COVERING problem.
Thus, Problem~\ref{prob:no_key_k} is coNP-complete, and this problem is polynomially reducible to PASSKEY, so PASSKEY is coNP-complete.\smartqed
\end{proof}

\textbf{Corollary.} Passkeys cannot be enumerated in total polynomial time unless P = coNP.

Thus, enumeration of the keys or generators are computationally hard. However, knowing the closure level to which a closed itemset belongs, the problem of enumerating passkeys becomes polynomial since for a closed itemset $B \in \mathcal{C}_k$, the number of all possible passkeys is given by
$$\binom{|M|}{k} = \frac{|M|!}{k!(|M|-k)!} = \frac{(|M| - k + 1) \cdot \ldots \cdot |M|}{k!} = \frac{|M|^{k}}{k!} \in O(|M|^{k}).$$ 
Thus, if $k$ is constant, one can identify all passkeys in polynomial time.
We also emphasize that the maximal possible $k$ depends on $M$, i.e., $k \leq |M|$, however for the concepts from the first levels, where $k$ is small, we can neglect the fact that $k$ may be proportional to $M$.
Moreover, a closed itemset $B$ from the $k$-th level provides the implications with the shortest antecedent $k$, i.e., size of the passkey, and largest consequent $|B|-k$.
Thus the first levels provide the implications that are commonly considered to be the best for interpretation purposes~\citep{bastide2000mining}.

Moreover, passkeys may provide much shorter representation of equivalence classes than keys.

\begin{proposition}
Consider the context $\mathbb{K} = (G,M,I)$, where $G= \{g_0, \ldots, g_n\}$, $M = \{m_0, \ldots, m_{2n}\}$, $g_0' = M$, and $g_i' = M\setminus \{m_0, m_i, m_{n+i}\}$.
The number of keys is $2^n$, whereas there is a unique passkey.
\label{prop:one_vs_exp}
\end{proposition}

An example of such a context for $n = 3$ is given in Fig.~\ref{fig:passkey_vs_key}. The elements of the equivalence class of formal concept $(g_0, M)$, which has a unique passkey and exponentially many keys, are given on the right.

\begin{figure}[h]
    \begin{minipage}{.36\textwidth}
        \centering
        \setlength{\tabcolsep}{1.5pt}
        \begin{tabular}{l|lllllll} \toprule
             &$m_0$&$m_1$&$m_2$&$m_3$&$m_4$&$m_5$&$m_6$\\ \midrule
            $g_0$ &\X &\X& \X& \X& \X& \X&\X \\
            $g_1$  &  &  &\X& \X &  & \X& \X\\
            $g_2$  &  &\X&  & \X& \X&  & \X \\
            $g_3$  &  &\X&\X &  & \X& \X &  \\ \bottomrule
        \end{tabular}
    \end{minipage}
    \begin{minipage}{.58\textwidth}
        \setlength{\tabcolsep}{1.5pt}
        \begin{tabular}{ll}  \toprule
            set name& members of the set \\ \midrule
            $pKey(M)$&  $m_0$\\\hline 
            $Key(M) \setminus pKey(M)$& \begin{tabular}[c]{@{}l@{}} $m_1m_2m_3$, $m_1m_2m_6$, $m_1m_5m_3$,\\ $m_1m_5m_6$, $m_4m_2m_3$, $m_4m_2m_6$,\\ $m_4m_5m_6$, $m_4m_5m_6$\end{tabular} \\\hline 
            $Equiv(M) \setminus Key(M)$& \begin{tabular}[c]{@{}l@{}} $m_0m_1m_2m_3$, $m_0m_1m_2m_6$,\\ $m_0m_1m_5m_3$, $m_0m_1m_5m_6$, \\$m_0m_4m_2m_3$, $m_0m_4m_2m_6$,\\ $m_0m_4m_5m_6$, $m_0m_4m_5m_6$\end{tabular} \\ \bottomrule
    \end{tabular}
    \end{minipage}
    \caption{An example of the context with the closed itemset $M$ having exponentially many keys and a unique passkey}
    \label{fig:passkey_vs_key}
\end{figure}

\begin{proof}
Let us consider the number of combinations of attributes that make a contranominal-scale subcontext of size $n$.
The set of objects of this subcontext is $\{g_1, \ldots, g_n\}$.
Each attribute with the extent $ \{g_1, \ldots, g_n\} \setminus \{g_i\}$ can be chosen in 2 different ways, i.e. either $m_i$ or $m_{i + n}$.
Thus, there exist $2^n$ ways to create the contranominal-scale subcontexts of size $n$ with attribute extents $\{g_1, g_2, \ldots, g_n\} \setminus \{g_i \}$, $i = 1, \ldots, n$.
All these $2^{n}$ attribute combinations have the same extent $\{g_0\}$ in the original context $\mathbb{K}$. In this context $\mathbb{K}$, there is another attribute $m_0$ having the same extent $\{g_0\}$. It means that the equivalence class $Equiv(M)$ has $2^n + 1$ keys, namely, $\{ m_0\} \cup \{m_{1 + k_1}m_{2 + k_2} \ldots m_{n + k_n} \mid k_i \in \{0, n\}, i = 1, \ldots, n\}$, and only one passkey $m_0$. \smartqed
\end{proof}

\section{Experiments}\label{sec:experiments}
\subsection{Characteristics of the datasets}

In this section we study some properties of the closure structure on real-world datasets. We use datasets from the LUCS-KDD repository~\citep{coenen2003}, their parameters are given in Table~\ref{tab:dataset}. The number of attributes in original/binarized data is given in columns ``\#attributes''/``$|M|$'', respectively. The relative number of ``1''s in binary data is given in column
``density''. We do not use the class labels of objects to compute the closure structure, but we use them to evaluate itemsets in Section~\ref{ssec:meaningfulness_f1}. The number of classes is given in column ``\#classes''.

\begin{table}[t]
    \caption{Parameters of the datasets}
    \label{tab:dataset}
    \centering
    %\begin{tabular}{l|ccccccr}
%    \toprule
%    name &$|G|$&$|M|$& density & \#classes &\#attributes& $N_c^{min}$ & \#concepts \\
%    \midrule
%    adult & 48842 & 95 & 0.15 & 2 & 14 & 12 & 359141 \\
%    auto & 205 & 129 & 0.19 & 6 & 8 & 8 & 57789.0 \\
%    breast & 699 & 14 & 0.64 & 2 & 10 & 6 & 361 \\
%    car\_evaluation & 1728 & 21 & 0.29 & 4 & 6 & 6 & 8000 \\
%    chess\_kr\_k & 28056 & 40 & 0.15 & 18 & 6 & 6 & 84636 \\
%    cylinder\_bands & 540 & 120 & 0.28 & 2 & 39 & 19 & 39829537 \\
%    dermatology & 366 & 43 & 0.28 & 6 & 33 & 9 & 14152 \\
%    ecoli & 327 & 24 & 0.29 & 5 & 8 & 6 & 425 \\
%    glass & 214 & 40 & 0.22 & 6 & 10 & 8 & 3246 \\
%    heart-disease & 303 & 45 & 0.29 & 5 & 0 & 9 & 25539 \\
%    hepatitis & 155 & 50 & 0.36 & 2 & 19 & 11 & 144871 \\
%    horse\_colic & 368 & 81 & 0.21 & 2 & 27 & 9 & 173866 \\
%    iris & 150 & 16 & 0.25 & 3 & 4 & 4 & 120.0 \\
%    led7 & 3200 & 14 & 0.5 & 10 & 7 & 7 & 1951.0 \\
%    mushroom & 8124 & 88 & 0.25 & 2 & 22 & 10 & 181945 \\
%    nursery & 12960 & 27 & 0.3 & 5 & 8 & 8 & 115200 \\
%    page\_blocks & 5473 & 39 & 0.26 & 5 & 10 & 8 & 723 \\
%    pen\_digits & 10992 & 76 & 0.21 & 10 & 16 & 11 & 3605507 \\
%    pima & 768 & 36 & 0.22 & 2 & 0 & 8 & 1626 \\
%    soybean & 683 & 99 & 0.32 & 19 & 35 & 13 & 2874252 \\
%    tic\_tac\_toe & 958 & 27 & 0.33 & 2 & 9 & 7 & 42712 \\
%    wine & 178 & 65 & 0.2 & 3 & 13 & 7 & 13229 \\
%    zoo & 101 & 35 & 0.46 & 7 & 17 & 7 & 4570\\
%    \bottomrule
%\end{tabular}
\setlength{\tabcolsep}{2.1pt}
\begin{tabular}{l|cccccc|l} \toprule
name &$|G|$&$|M|$&density&\#classes&\#attr.&\begin{tabular}[c]{@{}c@{}}closure\\index\\$CI$\end{tabular}& $2^{CI} \leq |\mathcal{C}| \leq 2^{\min(|M|,|G|)}$ \\\midrule
adult & 48842 & 95 & 0.15 & 2 & 14 & 12 &4.1e+3 $\leq$ 3.6e+5 $\leq$ 4.0e+28 \\
auto & 205 & 129 & 0.19 & 6 &\underline{\textbf{8}}&\underline{\textbf{8}}& 2.6e+2 $\leq$ 5.8e+4 $\leq$ 6.8e+38 \\
breast & 699 & 14 & 0.64 & 2 & 10 & 6 & 6.4e+1 $\leq$ 3.6e+2 $\leq$ 1.6e+4 \\
car\_evaluation & 1728 & 21 & 0.29 & 4 &\underline{\textbf{6}}&\underline{\textbf{6}}& 6.4e+1 $\leq$ 8.0e+3 $\leq$ 2.1e+6 \\
chess\_kr\_k & 28056 & 40 & 0.15 & 18 &\underline{\textbf{6}}&\underline{\textbf{6}}& 6.4e+1 $\leq$ 8.5e+4 $\leq$ 1.1e+12 \\
cylinder\_bands & 540 & 120 & 0.28 & 2 & 39 & 19 & 5.2e+5 $\leq$ 4.0e+7 $\leq$ 1.3e+36 \\
dermatology & 366 & 43 & 0.28 & 6 & 33 & 9 & 5.1e+2 $\leq$ 1.4e+4 $\leq$ 8.8e+12 \\
ecoli & 327 & 24 & 0.29 & 5 & 8 & 6 & 6.4e+1 $\leq$ 4.3e+2 $\leq$ 1.7e+7 \\
glass & 214 & 40 & 0.22 & 6 & 10 & 8 & 2.6e+2 $\leq$ 3.2e+3 $\leq$ 1.1e+12 \\
heart-disease & 303 & 45 & 0.29 & 5 &  & 9 & 5.1e+2 $\leq$ 2.6e+4 $\leq$ 3.5e+13 \\
hepatitis & 155 & 50 & 0.36 & 2 & 19 & 11 & 2.0e+3 $\leq$ 1.4e+5 $\leq$ 1.1e+15 \\
horse\_colic & 368 & 81 & 0.21 & 2 & 27 & 9 & 5.1e+2 $\leq$ 1.7e+5 $\leq$ 2.4e+24 \\
ionosphere & 351 & 155 & 0.22 & 2 & 34 & 17 & 1.3e+5 $\leq$ 2.3e+7 $\leq$ 4.6e+46 \\
iris & 150 & 16 & 0.25 & 3 &\underline{\textbf{4}}&\underline{\textbf{4}}& 1.6e+1 $\leq$ 1.2e+2 $\leq$ 6.6e+4 \\
led7 & 3200 & 14 & 0.50 & 10 &\underline{\textbf{7}}&\underline{\textbf{7}}& 1.3e+2 $\leq$ 2.0e+3 $\leq$ 1.6e+4 \\
mushroom & 8124 & 88 & 0.25 & 2 & 22 & 10 & 1.0e+3 $\leq$ 1.8e+5 $\leq$ 3.1e+26 \\
nursery & 12960 & 27 & 0.30 & 5 &\underline{\textbf{8}}&\underline{\textbf{8}}& 2.6e+2 $\leq$ 1.2e+5 $\leq$ 1.3e+8 \\
page\_blocks & 5473 & 39 & 0.26 & 5 & 10 & 8 & 2.6e+2 $\leq$ 7.2e+2 $\leq$ 5.5e+11 \\
pen\_digits & 10992 & 76 & 0.21 & 10 & 16 & 11 & 2.0e+3 $\leq$ 3.6e+6 $\leq$ 7.6e+22 \\
pima & 768 & 36 & 0.22 & 2 & & 8 & 2.6e+2 $\leq$ 1.6e+3 $\leq$ 6.9e+10 \\
soybean & 683 & 99 & 0.32 & 19 & 35 & 13 & 8.2e+3 $\leq$ 2.9e+6 $\leq$ 6.3e+29 \\
tic\_tac\_toe & 958 & 27 & 0.33 & 2 & 9 & 7 & 1.3e+2 $\leq$ 4.3e+4 $\leq$ 1.3e+8 \\
wine & 178 & 65 & 0.20 & 3 & 13 & 7 & 1.3e+2 $\leq$ 1.3e+4 $\leq$ 3.7e+19 \\
zoo & 101 & 35 & 0.46 & 7 & 17 & 7 & 1.3e+2 $\leq$ 4.6e+3 $\leq$ 3.4e+10 \\\bottomrule
\end{tabular}
\end{table}

First of all we study the closure levels of data and show how they are related to the size of the datasets and the estimated number of concepts. Column $CI$ contains the total number of closure levels. These values are quite low w.r.t. the total number of attributes $|M|$, but usually close to the number of attributes in the original data (i.e., ``\#attributes''). In our experiments we observe several datasets, where the closure index $CI$ is equal to the number of attributes in the original data, namely ``auto'', ``car evaluation'', ``chess kr k'', ``iris'', ``led7'', ``nursery''. Thus, from Proposition~\ref{prop:ci_binarized}, there exist at least $CI$ object descriptions such that each pair of them differ in only one pair of numerical value (where numerical values are distinguished up to binarization intervals). For ``car evaluation'' and ``nursery'' the number of closed itemsets of the highest complexity is equal to the number of objects, i.e., for each object description and each its numerical value there exists another object description that differs from it in a pair of attributes.

%and the number of such itemsets is equal to 1728, 19052, 2 and 12960, respectively. 
The observed variability of attribute values may be caused by a very fine-grained (too detailed) discretization. This variability may hamper itemset mining and can serve as an indicator for choosing a coarser binarization. In this paper we do not study an appropriate binarization strategy, but use the one generally used in related studies~\citep{coenen2003}.

However, high values of $CI$ may be caused not only by a too fine-grained discretization of real-valued attributes, but also the ``probabilistic'' nature of data, i.e., when the relation between object $g$ and attribute $m$ may appear with a certain probability, e.g., in the market basket analysis an item (dis)appears in a transaction with a certain probability. We consider this kind of data in the Appendix.

In Section~\ref{ssec:lower_bound_lattice_size} we mentioned that $CI$ may be used to compute a lower bound on the size of the concept lattice, i.e., the total number of concepts/closed itemsets is bounded from below by $2^{CI}$. In the last column of Table~\ref{tab:dataset} we list this lower bound together with the actual number of closed itemsets $\mathcal{C}$ and the upper bound $2^{\min(|M|,|G|)}$. Our experiments clearly show that the lower bound is much closer than the upper bound to the actual value of the lattice size.

%However, we derived the lower bound knowing the whole closure structure, thus all lower bounds are obtained post factum. An interesting research question here is whether we can compute efficiently the closure index without computing the whole structure, probably considering a dataset w.r.t. object descriptions. 

In this section we considered the size of the closure structure w.r.t. the dataset size and discussed some indicators of noise or too detailed discretization. In the next section we consider computational aspects of the closure structures and discuss the performance of the GDPM algorithm.

\subsection{Computational performance} 

As it was mentioned in Section~\ref{ssec:computing_cs_keys_min_keys}, when using passkeys to generate closed itemsets, we cannot rely on the canonicity test to ensure uniqueness in generating closed itemsets. Thus, the main limitations of the GDPM algorithm are that (i) it may take a lot of space for a trie to store all generated closed itemsets, (ii) the traversal strategy may also be redundant because we cannot use an efficient strategy to check uniqueness of the generated closed itemset. To estimate the impact of these limitations we study the execution time and the number of nodes in the tries. The results are reported in Table~\ref{tab:me_f}.

\begin{table}[t]
    \caption{Performance of GDPM-INT/EXT and key enumerator Talky-G}
    \label{tab:me_f}
    \centering
    \setlength{\tabcolsep}{2pt}
\begin{tabular}{l|ccc|c|cc|cc|c}\toprule
\multirow{2}{*}{name} & \multirow{2}{*}{$|G|$} & \multirow{2}{*}{$|M|$} & \multirow{2}{*}{\#attr.} & \multirow{2}{*}{\begin{tabular}[c]{@{}c@{}}\#closed\\ itemsets\end{tabular}}& \multicolumn{2}{c}{\#nodes in $\mathcal{T}$} & \multicolumn{3}{c}{runtime, sec} \\
 &  &  &  &  & INT & EXT & INT & EXT & \begin{tabular}[c]{@{}c@{}}Charm +\\Talky-G\end{tabular}\\\midrule
adult & 48842 & 95 & 14 & 359141 & 359141 & 338906& 984 & 837 & 13\\
auto & 205 & 129 & 8 & 57789 & 57390 & 48834 & 6 & 2 & 5\\
breast & 699 & 14 & 10 & 361 & 361 & 354 & 0 & 0 & 1\\
car evaluation & 1728 & 21 & 6 & 8000 & 4875 & 6901 & 1 & 0 & 1\\
chess kr k & 28056 & 40 & 6 & 84636 & 54761 & 68869 & 146 & 148 & 5\\
cylinder bands & 540 & 120 & 39 & 39829537 & 39829537 & 37007525 & 2404 & 4010 & -$^\ast$\\
dermatology & 366 & 43 & 33 & 14152 & 10506 & 12172 & 1 & 1 & 1\\
ecoli & 327 & 24 & 8 & 425 & 425 & 373 & 0 & 0 & 1\\
glass & 214 & 40 & 10 & 3246 & 2393 & 2757 & 0 & 0 & 1\\
heart-disease & 303 & 45 & 0 & 25539 & 19388 & 21869 & 1 & 1 & 1\\
hepatitis & 155 & 50 & 19 & 144871 & 88039 & 122277 & 2 & 1 & 5\\
horse colic & 368 & 81 & 27 & 173866 & 138075 & 138075 & 11 & 9 & 10\\
ionosphere & 351 & 155 & 34 & 23202541 & 23202541 & 20354049 & 2467 & 2585 &12628 \\
iris & 150 & 16 & 4 & 120 & 80 & 111 & 0 & 0 & 1\\
led7 & 3200 & 14 & 7 & 1951 & 1012 & 1931 & 0 & 0 & 1\\
mushroom & 8124 & 88 & 22 & 181945 & 181945 & 171590 & 164 & 121 & 8\\
nursery & 12960 & 27 & 8 & 115200 & 76800 & 103351 & 46 & 51 & 4\\
page blocks & 5473 & 39 & 10 & 723 & 451 & 702 & 0 & 3 & 1\\
pen digits & 10992 & 76 & 16 & 3605507 & 2423123 & 3236109 & 12863 & 3398 & 123\\
pima & 768 & 36 & 0 & 1626 & 1105 & 1444 & 0 & 0 & 1\\
soybean & 683 & 99 & 35 & 2874252 & 2726204 & 2640461 & 379 & 224 & 298\\
tic tac toe & 958 & 27 & 9 & 42712 & 25276 & 34513 & 1 & 1 & 1\\
wine & 178 & 65 & 13 & 13229 & 8518 & 10202 & 1 & 1 &  2 \\
zoo & 101 & 35 & 17 & 4570 & 2933 & 4125 & 0 & 0 &  1\\\midrule
\textbf{average} & 5239 & 57 & 15 & 2947747 & 2883953 & 2680313 & 812 & 475 & 376\\\toprule
\end{tabular}\\
\begin{flushleft}
$^\ast$ The algorithm crashes because of the lack of memory.
\end{flushleft}

\end{table}

Our experiments show that for small datasets, i.e. they contain less than or about 1K objects and a small number of attributes, the difference between the algorithms is negligible.
For larger datasets GDPM-EXT works faster when the number of attributes in original and binarized dataset is high and the number of objects is not too large.
However, both versions are not scalable, meaning that as the size of datasets increases, the computational time increases considerably.
We compare GDPM with the closest combination of the algorithms ``Charm+Talky-G'' implemented in the Coron data mining platform\footnote{http://coron.loria.fr/}.
The Charm algorithm enumerates closed itemsets while Talky-G~\citep{szathmary2014fast} enumerates all keys.
The experiments show that GDPM works slower than the closest in functionality combination of algorithms.
Actually, these results leave space for further study and improvement of the GDPM algorithm, in particular, by using hash tables and vertical representation of datasets.
It is also important to notice that GDPM and ``Charm+Talky-G'' perform different tasks.
% In Section~\ref{ssec:computing_min_keys} we discuss the corresponding complexity classes: enumeration of passkeys is \#coNP-complete, while enumeration of keys is \#P-complete.

Below we consider more deeply the behavior of algorithms GDPM-EXT and GDPM-INT in comparing the execution time and the size of the tries for the datasets where we can observe the difference between the algorithms. 
The level-wise execution time is reported in Fig.~\ref{fig:time_me_f}.
As it can be seen, the first and last levels are fast to compute. The most time-consuming levels are located in the middle of the closure structure.
For example, for ``adult'' levels 1-3 and 9-12 are computed in no more than 10 seconds.
While computing level 5 takes 319 and 233 seconds for GDPM-INT and  GDPM-EXT, respectively.
Almost for all datasets GDPM-EXT works faster than GDPM-INT.
We observe the largest difference in the running time for ``pen digits'', where computing level 6 took 4855 and 1222 seconds for GDPM-INT and GDPM-EXT, respectively.
However, for some datasets GDPM-INT works faster, e.g., for ``cylinder bands'' computing level 10 takes 874 seconds for the GDPM-EXT algorithm and only 316 seconds for GDPM-INT.
Moreover, the running time curve of GDPM-INT is flatter and does not have a pronounced peak.

\begin{figure}
    \centering
    \includegraphics[width = 1.\textwidth]{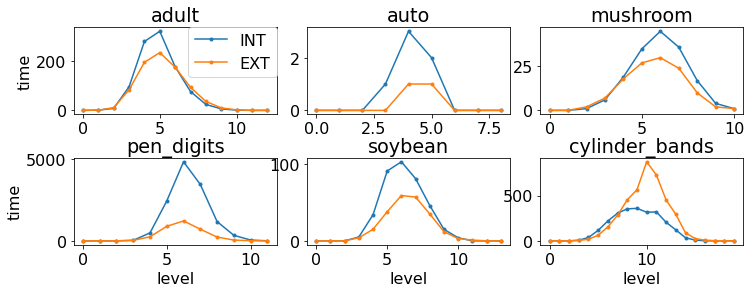}
    \caption{Level-wise running time of GDPM-INT / EXT for some datasets, the level number and time in seconds are given in axes $x$ and $y$, respectively. Level 0 contains only itemset $\emptyset$ in the root}
    \label{fig:time_me_f}
%\end{figure}
%
%\begin{figure}
    \centering
    \includegraphics[width = 1.\textwidth]{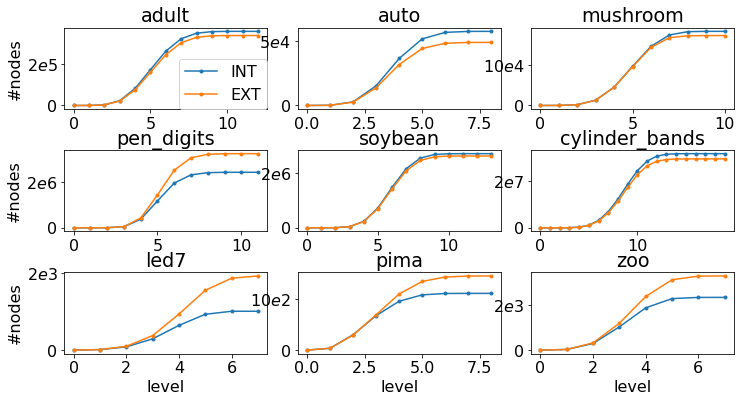}
    \caption{Growing of the tries in the number of nodes, the level number, and the trie size in nodes are given in axes $x$ and $y$, respectively}
    \label{fig:n_nodes_me_f}
\end{figure}

% As it was mentioned in Section~\ref{ssec:complexity_gdpm}, GDPM has exponential worst-case space complexity.
The trie that stores the generating closed itemsets grows at each iteration.
Fig.~\ref{fig:n_nodes_me_f} shows how the tries are growing with each new closure level.
We also show the evolution of the size of tries for the datasets where this difference is most distinguishable, namely ``led7'', ``pima'' and ``zoo''.
For the remaining datasets the sizes of tries that store closed itemsets and their extents are quite close one to another.
The total sizes of tries, as well as the corresponding running times for all datasets are reported in Table~\ref{tab:me_f}.

We emphasize that the size of the trie does not affect a lot the runtime of the duplicate test, i.e., checking if a closed itemset is generated for the first time, which takes at most $O(|M|)$ and $O(|G|)$ time for GDPM-INT and GDPM-EXT, respectively. 

The analysis of the algorithms shows that both algorithms compute first levels fast using not much space.
However, it remains to study whether we can stop computation at the first levels and whether the itemsets from the first levels are useful.
We address these questions in the next sections.

\subsection{Data topology or frequency distribution within levels}

The common approach to frequent itemset mining consists in gradually mining itemsets of decreasing frequency. For frequency-based enumeration methods infrequent itemsets are the last to be computed.
As we discussed above, it results in the need for computing exponentially many itemsets.
However, not all frequent itemsets are interesting and there are not that many infrequent itemsets that are worth to be considered.

In this section we study the frequency distribution of itemsets within closure levels.
We take the number of itemsets at level $k$ as 100\% and compute the ratio of itemsets within 5 frequency bins: [0, 0.2), [0.2, 0.4), [0.4, 0.6), [0.6, 0.8), [0.8, 1.0].
In Fig.~\ref{fig:frequency_disctribution} we report the distribution within these frequency bins. For example, the first figure shows that ``breast'' dataset has 6 levels, the largest one is the 4th, since it accounts for 34.4\% of all closed itemsets.
The first level contains itemsets of all frequency bins, while the last two levels only contain itemsets of frequency at most 0.4, i.e., [0, 0.2) and [0.2, 0.4).

\begin{figure}[t]
    \centering
    \includegraphics[width = 1.\textwidth]{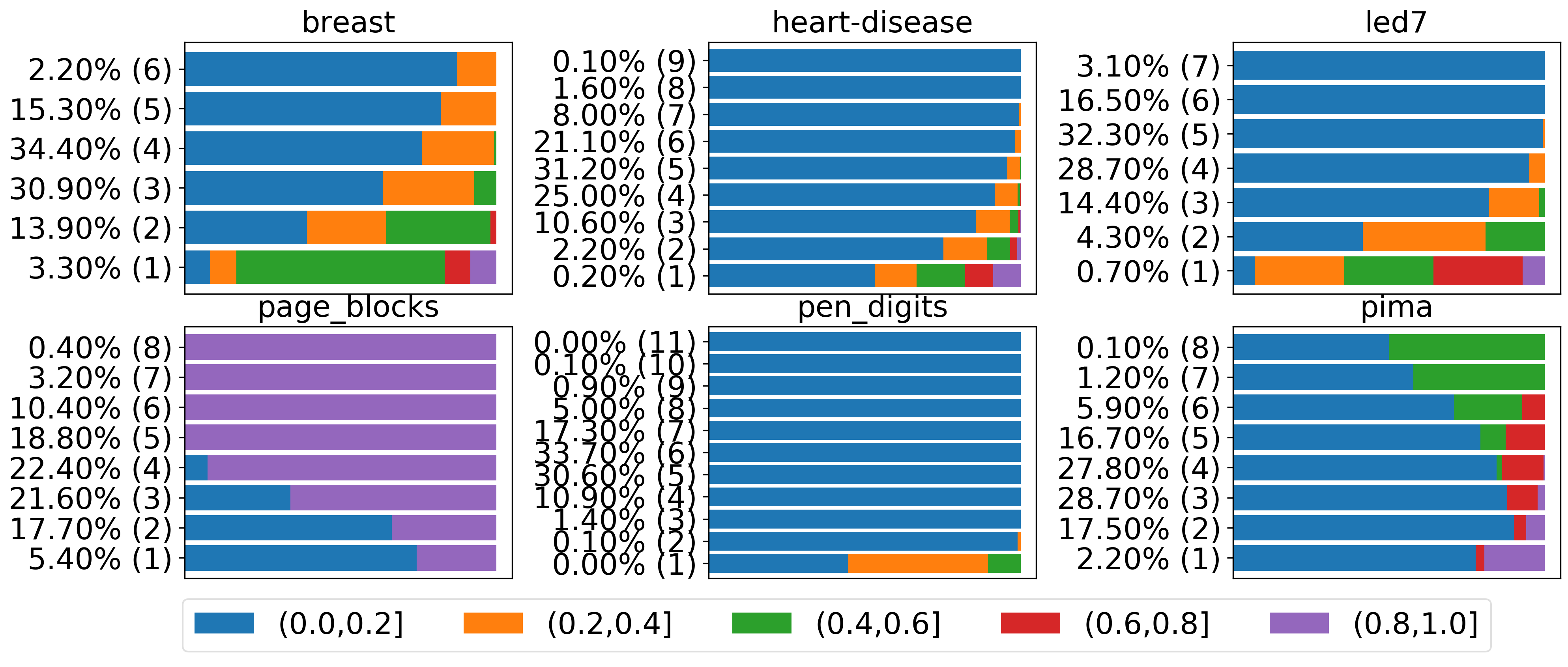}
    \caption{Distribution of closed itemsets within 5 frequency bins by the closure levels. The horizontal bars represent closure levels, the level number is given in parentheses. The rightmost value is the percentage of closed itemsets $|\mathcal{C}_{k}|/|\mathcal{C}|$. The width of the bar is proportional to the ratio of $\mathcal{C}_{k}$ of frequency $(v_1, v_2]$ among $\mathcal{C}_k$.}
    \label{fig:frequency_disctribution}

    \centering
    \includegraphics[width = 1.\textwidth]{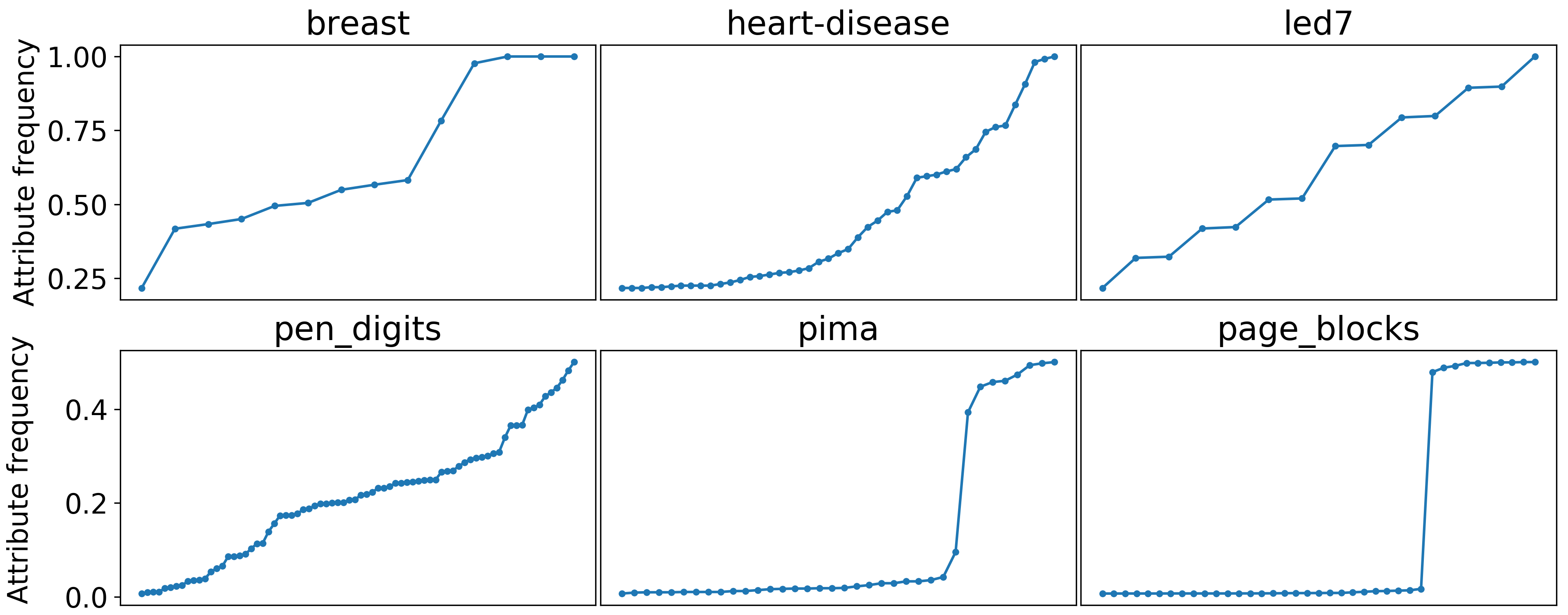}
    \caption{Frequency of attributes in ascending order for some datasets.}
    \label{fig:attribute_frequecy}
\end{figure}

Almost in all datasets we observe a quite similar behavior: the proportion of frequent itemsets decreases with the closure levels.
Actually, the first level contains a relatively small ratio of infrequent itemsets and a quite large ratio of frequent ones. In the next levels the ratio of frequent itemsets decreases, while the ratio of infrequent ones increases. 

For example, for ``breast'' dataset, closed itemsets of frequency in $(0.8, 1.0]$ (violet) are present only in the 1st level, where they constitute 8.3\% of all closed itemsets of $\mathcal{C}_1$.
The closed itemsets of frequency in $(0.6, 0.8]$ (red) account for 8.3\% of $\mathcal{C}_1$ and 2\% of $\mathcal{C}_2$.
The closed itemsets of frequency $(0.4, 0.6]$ (green) account for 67\% of $\mathcal{C}_1$, 33\% of $\mathcal{C}_2$, 7.1\% of $\mathcal{C}_3$, and 0.8\% of $\mathcal{C}_4$.
The less frequent ones (frequency below 0.2) constitute 8.3\% of $\mathcal{C}_1$ and 87.5\% of the last level $\mathcal{C}_6$. 

However, some datasets, e.g., ``page blocks'' and ``pima,'' show a different behaviour: the rate of frequent itemsets increases with the closure level.
This is typical for the datasets that have two types of attributes, e.g., of very low and very high frequency.
To illustrate it, we show in Fig.~\ref{fig:attribute_frequecy} the average frequency of the attributes of the datasets given in Fig.~\ref{fig:frequency_disctribution}.
We can see that the attribute frequency distributions of ``page blocks'' and ``pima'' differ from the others, i.e., there exist two groups of attributes of very contrasted frequencies.
Thus, the frequency distribution of attributes shows the itemset frequency distributions within the closure levels.

For the majority of datasets, the first closure levels contain itemsets of diverse frequencies, thus by generating only a \emph{polynomial} number of itemsets we may obtain itemsets of quite different frequencies.% What is also important is that the itemsets of the first levels have the simplest short lossless descriptions by their passkeys. 
 
Moreover, the frequency distributions by levels, as shown in Fig.~\ref{fig:frequency_disctribution}, provide a representation of the ``topology'' of the dataset in terms of closed sets.
For example, for ``breast'' dataset we may infer that we need at most 6 items to describe any closed itemset.
Moreover, we may use only itemsets of size 1 or 2 to concisely describe itemset of frequency above 0.6.
The closure structure provides the number of closed itemsets of a given frequency that can be represented without loss by passkeys of size $k$.
For example, the first four levels of ``breast'' dataset contain 8 (66.7\%), 17 (33.3\%), 8 (7.1\%) and 1 (0.7\%) closed itemsets of frequency $(0.4, 0.6]$.
Thus the majority of these frequent closed itemsets can be succinctly represented by passkeys of size 2 or 3.

From the experiments we can see that the first levels contain the itemsets with the most ``diverse'' frequency ranges. However, it is not clear how these itemsets cover the data in a satisfying way and this is what we study in the next section.

\subsection{Coverage and overlaps}

Usually, we want that the generated concepts cover data in the most uniform way, i.e., that the whole dataset is completely covered and that each entry is covered by roughly the same number of itemsets.
We evaluate the level-wise distribution of the closed itemsets by measuring the coverage and overlapping ratios.

Let $\mathbb{K} = (G, M, I)$ be  a formal context and $\mathcal{C}$ be a set of formal concepts, then the coverage ratio is defined as the ratio of non-empty entries in the context covered by the concepts from $\mathcal{C}$.
The overlapping ratio is defined as the average number of itemsets that cover a cell in the data.
Since the overlapping ratio is computed for the entries covered by at least one itemset, it may take any value from 1 to the number of itemsets.

Fig.~\ref{fig:coverage_overlaps} shows the coverage rate and overlapping ratios by levels for some datasets. The itemsets from the first level always cover entirely the whole dataset. Then, the coverage rate steadily decreases with each new level. For example, the coverage rates of the first levels of ``pima'' dataset are close to 1, at level 5 the coverage rate drops to 0.98, at level 7 it decreases to 0.82, and at level 8 it drops to 0.56. 

The overlapping ratio --the average number of itemsets that cover an entry of the dataset-- may vary a lot within a single level, but usually the closure levels in the middle have the largest average number of itemsets per entry.
For example, the overlapping ratio of ``pima'' may vary a lot from level to level.
The overlapping ratio at the 2nd level is low  (equals to 7.32 $\pm$ 1.24), at the next levels it increases and achieves its maximum at the level 4 (30.27 $\pm$ 8.09).
The latter means that each entry of a dataset is covered on average by 30 itemsets from the level 4.

Almost for all datasets we observe quite the same behavior: usually the coverage rate starts to decrease at the levels having the highest overlapping ratio.
We also observe that with quite the same overlapping ratio the itemsets from the first levels (before the peak on the overlapping curve) cover much larger fragments of data than those from the last levels.
The latter means that, compared with the first levels, the description of the last levels is more focused on very particular data fragments. 

\begin{figure}
    \centering
    \includegraphics[width = 1.\textwidth]{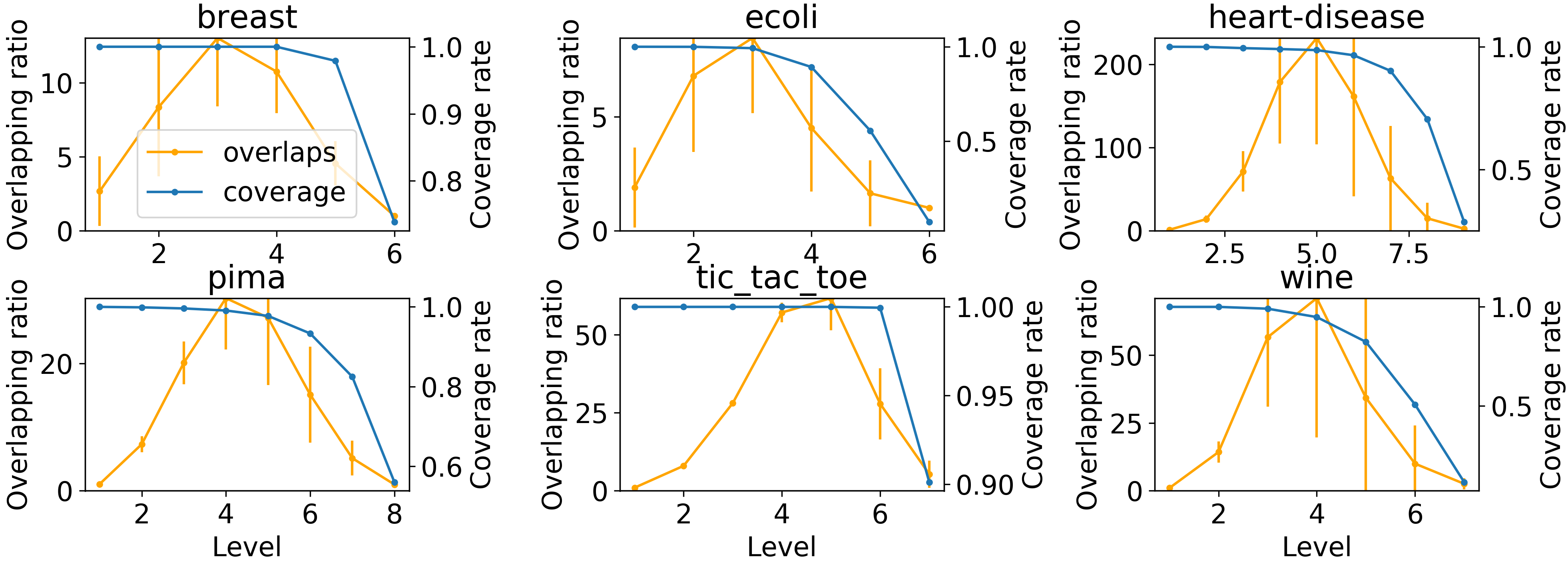}
    \caption{Coverage rate and mean overlapping ratio of the cells by levels. The vertical lines show the standard deviation of the overlapping ratio over cells within a given level.}
    \label{fig:coverage_overlaps}
\end{figure}

However, based on these experiments, we should now investigate whether the itemsets from the first levels are more useful than those from the last levels. 

\subsection{Meaningfulness of concepts (w.r.t. F1 measure)}\label{ssec:meaningfulness_f1}

In the previous sections we considered the closure levels and showed that the first levels contain itemsets with more diverse frequencies, with a larger amount of frequent itemsets and a smaller number of infrequent ones. In Section~\ref{ssec:closure_structure} we also mentioned that the itemsets from the first levels can be replaced by passkeys from the same equivalence class, allowing for much more succinct descriptions.

In this section we study the usefulness of the generated closed itemsets and focus on the ability of itemsets to describe ``actual'' classes in data. To measure this ability of an itemset we study the average F1 measure by levels. We consider each closed itemset $B$ with its extent $A$ as well as the class labels that were not used before. For a formal concept $(A,B)$, where each object $g \in A$ has an associated class label $class(g) \in \mathcal{Y}$, the class label of an itemset is computed by the majority vote scheme, i.e. $class(B) = \arg\max_{y \in \mathcal{Y}} \{|\{ class(g) = y \}|\mid g \in A \}$. Then, F1 measure for $(A,B)$ is given by
$$ F1(A,B) = \frac{2tp}{2tp + fp + fn}$$ where $tp = |\{g \mid class(g) = class(B), g \in A\}|$, $fp = |A| - tp$, $fn = |\{g \mid class(g) = class(B), g \in G\setminus A\}$.

F1 measure is a harmonic mean of precision and recall, i.e., it takes into account how large is the ratio of true positives it describes and how many instances of the true class it covers. 
Since  the ``support'' of the concept may affect the value of F1 measure, we consider F1 measure within 10 frequency bins: (0, 0.1], (0.1, 0.2], \ldots, (0.9, 1.0]. The average values for some datasets are given in Fig.~\ref{fig:f1}.

\begin{figure}
    \centering
    \includegraphics[width = 1.\textwidth]{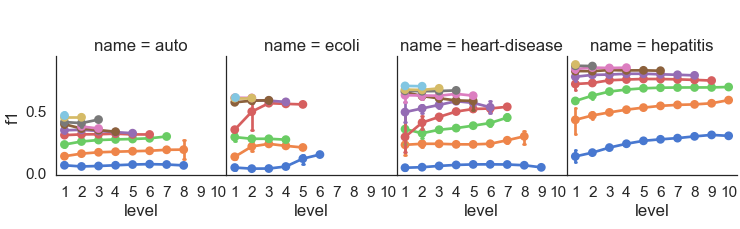}
    \includegraphics[width = 1.\textwidth]{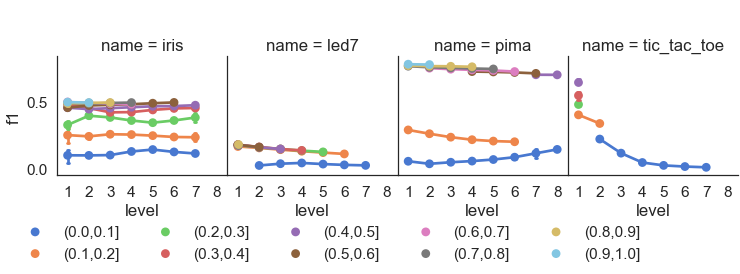}
    \caption{Average F1 measure within frequency bins}
    \label{fig:f1}
\end{figure}

In ``auto'' dataset, for example, the most frequent closed itemsets ($fr. \in (0.9, 1.0]$) are located only at the 1st level and have the largest average F1 value. As the frequency decreases the average F1 measure decreases as well. The smallest average F1 corresponds to the itemsets of frequency $(0, 0.1]$. This dependence of the average F1 values on frequency bins of itemsets remains the same over the levels. For example, at level 7 we observe the highest values for itemsets with frequency in the interval $(0.2, 0.3]$ (the largest frequency at the level). Almost for all datasets the values of the average F1 measure remain quite similar for itemsets within a frequency bin over all closure levels. %However, since the rate of infrequent itemsets increases with each level, the average F1 measure (over itemsets of all frequencies) decreases with the closure levels (see Fig.~\ref{fig:avgF1})

%\begin{figure}
%    \centering
%    \includegraphics[width = 1.\textwidth]{pics/avgF1.png}
%    \caption{Average F1 measure.}
%    \label{fig:avgF1}
%\end{figure}
%\begin{figure}
%    \centering
%    \includegraphics[width = 0.9\textwidth]{pics/size_1.png}
%    \includegraphics[width = 0.9\textwidth]{pics/size_2.png}
%    \caption{Caption}
%    \label{fig:avgF1}
%\end{figure}

Thus, the quality of itemsets evaluated w.r.t. F1 measure does not change over the levels and depends rather on the frequency of itemsets. 

\section{Discussion and Conclusion}
\label{sec:conclusion}

In this paper we introduced the ``closure structure'' of a dataset, with computational aspects, theoretical and practical properties.
This closure structure provides a view of the ``complexity'' of the dataset and may be used for guiding the mining of the dataset.
We show theoretical results about the ordinal nature of the closure structure, and we state and prove complexity results about the computation of passkeys.
Moreover, the GDPM algorithm and its variations are able to compute the closure structure.
The application of GDPM over a collection of datasets shows a rather efficient behavior of the algorithm.
For synthesizing, the closure structure of a dataset brings the following benefits:

\begin{itemize}
\item
  The closure structure is determined by the closed itemsets and their equivalence classes.
  It is based on a partition related to the smallest keys --smallest itemsets in an equivalence class-- called passkeys and that can represent any closed subset of attributes and related objects without loss.
\item
  The closure structure determines a ``topology'' of a dataset and is organized around a set of levels depending on the size of the passkeys.
  This provides a distribution of the itemsets in terms of their associated passkeys.
  Moreover, the number of levels, called the closure index, can be used to measure the potential size of the pattern space, i.e. given by the minimal Boolean lattice included in the pattern space.
\item
  Given the closure structure and the closure levels, we may understand how itemsets are distributed by levels and know when to stop the search for interesting itemsets with a minimal loss.
  % can we qualify and quantify this loss?
  For example, we discovered that the first levels of the closure structure have better characteristics than the other levels.
  In particular, the passkeys of the first levels provide an almost complete covering of the whole dataset under study, meaning also that the first levels are the most interesting and the less redundant in terms of pattern mining.
  In addition, the first levels of the closure structure allow to derive the implications with the smallest antecedents and the maximal consequents, which are considered as the best rules to be discovered in a dataset.
\item
  Finally, we also studied how to use the closure structure with large or very large datasets using sampling.
  We showed that the closure structure is stable under the sampling of objects, provided that the set of attributes remains identical.
  Then, a combination of feature selection and object sampling would be a possible way of dealing with very large datasets.
\end{itemize}

This is the first time that the closure structure is studied on its own and not all aspects were investigated in this paper.
However, the proposed closure structure has some limitations that we aim to improve in future work:
the closure structure provides a very rich framework for understanding the content of a dataset and thus there are many directions for future work.
We sketch some of these future investigations below.

\begin{itemize}
\item
  Considering complex datasets such as numerical datasets, how to generalize the closure structure to the mining of numerical datasets of linked data for example?
  One possible direction to deal with such complex data is to use the formalism of pattern structures, which is an extension of FCA allowing to directly work on complex data such as numbers, sequences, trees, and graphs \citep{kaytoue2011mining}.
\item
  The closure structure is introduced and studied thanks to frequency which depends on the support of itemsets.
  Along with frequency, it could be very interesting to reuse other measures that would bring different and complementary views on the data, such as e.g. stability, lift, and MDL.
\item
  The closure structure is intended in this paper in terms of pattern mining, and more precisely in terms of itemset mining.
  The question of the generalization can be raised, for example how the closure structure could be reused in terms of investigations related to pattern mining, such as subgroup discovery, anytime approaches in pattern mining, and also, as already mentioned above, MDL-based pattern mining?
\item
  Many questions remain when we consider the practical point of view.
  Among these questions, how to define a possibly anytime strategy for pattern mining terms based on the closure structure, knowing in particular that the first levels are the most interesting.
  Moreover, this strategy could also guide the discovery of the most interesting implications and association rules, and take into account the relations with functional dependencies.
\item
  The GDPM algorithm should be improved in terms of computational efficiency --some of them are already present in Talky-G for example-- and as well adapting the algorithm to other datatypes such as numbers, RDF data, as already mentioned above.
  Finally, another important direction of investigation is to study the possibility of designing a depth-first version of GDPM with an adapted canonicity test, which may improve a lot the computation performance.
  % Regarding the first limitation, finding a strategy for traversing
  % the items search space, i.e., as the canonicity test or recursive
  % traversal where one keeps some stop-directions, may greatly help
  % to improve the computation performance.
\end{itemize}

\bibliography{refs}

\newpage
\appendix
\section{Case study} \label{sec:market_basket_analysis}

In this paper we discussed the importance of the closure structure in finding association rules. In this section we demonstrate how the closure structure may be used to obtain exact association rules (implications) in market basket analysis. We consider a dataset from Kaggle repository\footnote{https://www.kaggle.com/apmonisha08/market-basket-analysis}. It consists of 7835 objects and 167 attributes, the dataset density is 0.026. Frequencies of attributes are shown in Fig.~\ref{fig:freq_attributes} in ascending order. This frequency distribution is typical for such kind of data - there is a few very frequent items, e.g., `whole milk', `other vegetables', `soda', `rolls/buns', `yogurt',  `bottled water', `root vegetables', `shopping bags' and `tropical fruit', and a large amount of very infrequent items.

\begin{figure}[h!]
    \centering
    \begin{minipage}{.48\textwidth}
        \includegraphics[width = 1.\textwidth]{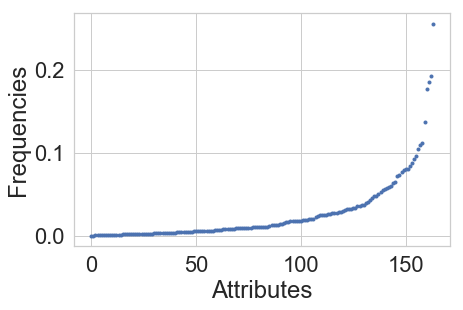}
    \end{minipage}
    \caption{Frequency of attributes in ascending order}
    \label{fig:freq_attributes}
\end{figure}

The standard approach for computing association rules consists of two stages: (i) enumerating all frequent (closed) itemsets and (ii) computing the rules by splitting each itemset into an antecedent and consequent in all possible ways. We use MLxtend library\footnote{https://pypi.org/project/mlxtend/} to generate association rules. The number of enumerated frequent itemsets is 637606, frequency threshold is 0.0005.

To compare the standard approach with one based on the closure structure, we consider rules with confidence equal to 1 (i.e., implications) and order them by lift. The top 5 rules are given in Table~\ref{tab:frequent_conf_0005}, they are very similar, e.g., the 2nd and 4th rules have the same antecedent and almost the same consequents. 

\begin{table}[h!]
    \caption{Top-5 association rules with confidence equal to 1 ordered by lift}
    \label{tab:frequent_conf_0005}
    \setlength{\tabcolsep}{2pt}    
\begin{tabular}{L{4.5cm}L{4.5cm}|c|c|c} \toprule
    antecedents & consequents &support&conf&lift\\ \midrule
    \nt{`frankfurter', `frozen fish',\\`whole milk', `pip fruit'}&`frozen meals', `other vegetables', `yogurt', `tropical fruit'&0.0005&1&1567.0\\\hline
    `grapes', `frozen fish'&`frozen meals', `pip fruit', `other vegetables',`tropical fruit', `whole milk'&0.0005&1&1119.3\\\hline
    \nt{`frozen meals', `other vegetables',\\`yogurt', `frankfurter'}&\nt{`tropical fruit', `frozen fish',\\ `whole milk', `pip fruit'}&0.0005&1&1119.3\\\hline
    `grapes', `frozen fish'&\nt{`frozen meals', `other vegetables',\\`tropical fruit', `pip fruit'}&0.0005&1&870.6\\\hline
    `grapes', `whole milk', `frozen fish'&\nt{`frozen meals', `other vegetables',\\`tropical fruit', `pip fruit'}&0.0005&1&870.6\\\bottomrule
%    'grapes', 'frozen fish'&\nt{'frozen meals', 'other vegetables',\\'whole milk', 'pip fruit'}&783.5\\
%    \nt{'frozen meals', 'other vegetables',\\'yogurt', 'frankfurter'}&\nt{'tropical fruit', 'frozen fish',\\ 'pip fruit'}&783.5\\
%    'tropical fruit', 'grapes', 'frozen fish'&\nt{'frozen meals', 'other vegetables',\\'whole milk', 'pip fruit'}&783.5\\
%    \nt{'frozen meals', 'yogurt', \\'pip fruit', 'other vegetables',\\'tropical fruit'}&\nt{'frankfurter', 'whole milk',\\'frozen fish'}&783.5\\
%    \nt{'frozen meals', 'yogurt', \\'other vegetables', 'frankfurter',\\'whole milk'}&\nt{'tropical fruit', 'frozen fish',\\'pip fruit'}&783.5\\
%    \nt{'domestic eggs', 'citrus fruit',\\'misc. beverages'}&'whole milk', 'canned fruit'&712.3\\
%    \nt{'frozen meals', 'other vegetables',\\'yogurt', 'frankfurter'}&\nt{'frozen fish', 'whole milk',\\'pip fruit'}&712.3\\
%    \nt{'frozen meals', 'yogurt',\\'other vegetables', 'tropical fruit',\\'frankfurter'}&\nt{'frozen fish', 'whole milk',\\'pip fruit'}&712.3\\
%    \nt{'frozen meals', 'yogurt',\\'whole milk', 'frozen fish'}&\nt{'other vegetables', 'tropical fruit',\\ 'frankfurter', 'pip fruit'}&712.3\\\bottomrule
    \end{tabular}
\end{table}

Because of the redundancy problem, it is hard to analyze the obtained rules and to get an overall idea about them. However, the closure structure provides an exhaustive set of non-redundant rules and provides a clear representation of them. 

Table~\ref{tab:level_summary} contains the description of the computed closure structure. Each row corresponds to a closure level, each column correspond to a size of closed itemsets. All the implications $X\rightarrow Y$ summarized in the $k$-th row of the table have the smallest possible antecedent (which is a passkey) among all the equivalent implications (i.e., the implications $X^\ast\rightarrow Y^\ast$ with confidence 1, such that $X^\ast \cup Y^\ast = X \cup Y$). The number of implications is computed up to equivalence classes, i.e., if there are several implications with minimum antecedents of size $k$ in an equivalence class, we count only.

A cell contains the number (and the average support) of itemsets of a certain size belonging to a certain level. For example, the 1st level contains 163 itemsets of size 1 (singletons) with the average support 211.4$\pm$304.1 and 1 itemset of size 5 having support 2. The 2nd closure level contains itemsets of the sizes from 2 to 14. The itemsets that are located in the diagonal (in gray) can be used to build implications since their consequents are empty. Analyzing the itemsets of size at least 3 of the 2nd level we may conclude that there exists 629 implication with the consequents of size 1 and average support 2.8$\pm$1.1, 264 implications with the consequents of size 2 and average support 2.3$\pm$0.6, etc. The largest consequents of the implications with antecedent 2 have size 10, 11, 12, their support is 2.

\begin{table}
    \caption{The most frequent implications obtained from the itemsets of the 2nd level}
    \label{tab:implications_2nd_level}
    \begin{tabular}{L{5cm}L{5cm}|c} \toprule
    antecedents & consequents & support \\ \midrule
    {`herbs', `long life bakery product'}&  `whole milk', `{root vegetables'} &  5 \\ \hline
    {`specialty chocolate', `semi-finished bread'}& {`pastry'}, `whole milk'&   5 \\  \hline
    `citrus fruit', `rubbing alcohol', & `{root vegetables}', `whole milk'& 4 \\\hline
    {`root vegetables', `rubbing alcohol'}& {`citrus fruit'}, `whole milk'& 4 \\ \hline
    `curd cheese', `hygiene articles' &'other vegetables', `yogurt'& 4 \\ \hline
    {`grapes', `frozen fish' }&{`frozen meals',} `whole milk', `other vegetables', {`pip fruit'}, `tropical fruit'& 4 \\ \hline
    `frozen fish', `soft cheese' &`yogurt', `root vegetables'& 4 \\ \hline
    `turkey', `hamburger meat'& `yogurt', `whole milk'& 4 \\ \hline
    `spread cheese', `hygiene articles'& `other vegetables', `yogurt' & 4 \\ \hline
    `sugar', `meat spreads' & `yogurt', `whole milk' & 4 \\ \hline
    {`waffles', `nuts/prunes'}& {`newspapers'}, `whole milk' &4 \\ \hline
    {`specialty fat', `pip fruit'}& {`fruit/vegetable juice'}, `whole milk' & 4 \\ \hline
    {`rice', `roll products'}& `other vegetables', {`fruit/vegetable juice'} & 4 \\ \hline
 %   `sliced cheese', `rice' & `whole milk', `other vegetables', `root vegetables', `yogurt', `tropical fruit' & 4 \\\bottomrule
    \end{tabular}
\end{table}

In Table~\ref{tab:implications_2nd_level} we list some implications obtained from itemsets from the 2nd level. Among the obtained implications, we can observe quite specific implications, e.g., if one buys `specialty chocolate' and `semi-finished bread' one buys `pastry' as well, or if one buys `waffles' and `nuts/prunes' one buys `newspapers' as well. 

One of the benefits of the proposed approach is that (i) we deal with easy interpretable association rules, i.e. those that have the smallest possible antecedent and the largest possible consequent, (ii) the closure structure gives the exact number of rules with a chosen sizes of antecedents and consequents, more over these rules are non-redundant. 

As we can seen, the proposed approach provides implications that can be more easily analyzed by experts. However, the proposed approach has some drawbacks: the implications are more suitable for ``deterministic'' data (see Section~\ref{sec:related_work}), where the closure structure may serve a good tool to obtain the data summary. 
For noisy datasets or datasets where the appearance of an item in a transaction (row) is rather probabilistic (e.g., an item appears in a shopping cart with a certain probability) computing the exact structure maybe be inappropriate because of its sensitivity to noise (a lot of noisy unstable itemsets may appear). Thus, an important direction of future work is to generalize the notion of closure structure for approximate tiles (biclusters), i.e., dense rectangles in data, to deal with ``approximate closures''. 

\begin{sidewaystable}
    \caption{Summary of the levels of the closure structure (for itemsets with support at least 1)}
    \label{tab:level_summary}
    %\rotatebox{90}{
\setlength{\tabcolsep}{4pt}
\adjustbox{max width=\textwidth}{%
\begin{tabular}{l|l|cccccccccccccccc}\toprule
 &  & 1 & 2 & 3 & 4 & 5 & 6 & 7 & 8 & 9 & 10 & 11 & 12 & 13 & 14 & 15 & 16 \\ \midrule
\multirow{2}{*}{1} & count &{\color[HTML]{9B9B9B}163}&  &  &  & 1 &  &  &  &  &  &  &  &  &  &  &  \\
 &sup.&{\color[HTML]{9B9B9B}211.4+304.1}&  &  &  & 2.0+0.0 &  &  &  &  &  &  &  &  &  &  &  \\\hline
\multirow{2}{*}{2} & count &  &{\color[HTML]{9B9B9B}5637}& 629 & 264& 136& 71& 26 & 17 & 14 & 3 & 2 & 1& 1 & 1 &  &  \\
 &sup.&  &{\color[HTML]{9B9B9B}18.3+31.6}& 2.8+1.1 & 2.3+0.6 & 2.1+0.3 & 2.1+0.4 & 2.1+0.4 & 2.1+0.2 & 2.0+0.0 & 2.0+0.0 & 2.0+0.0 & 2.0+0.0 & 2.0+0.0 & 2.0+0.0 &  &  \\\hline
\multirow{2}{*}{3} & count &  &  &{\color[HTML]{9B9B9B}27051}& 7655& 3409& 1483 & 677& 294 & 108 & 57 & 21 & 8 & 3 & 3 & 1 & 1 \\
 &sup.&  &  &{\color[HTML]{9B9B9B}6.8+7.7}& 2.8+1.1 & 2.4+0.7 & 2.2+0.5 & 2.1+0.4 & 2.1+0.3 & 2.1+0.2 & 2.1+0.3 & 2.0+0.0 & 2.0+0.0 & 2.0+0.0 & 2.0+0.0 & 2.0+0.0 & 2.0+0.0 \\\hline
\multirow{2}{*}{4} & count &  &  &  &{\color[HTML]{9B9B9B}28841}& 10115& 3526& 1211 & 376& 84 & 24 & 6 & 1 &  &  &  &  \\
 &sup.&  &  &  &{\color[HTML]{9B9B9B}4.6+3.3}& 2.9+1.1 & 2.5+0.8 & 2.3+0.6 & 2.2+0.4 & 2.1+0.3 & 2.1+0.3 & 2.0+0.0 & 2.0+0.0 &  &  &  &  \\\hline
\multirow{2}{*}{5} & count &  &  &  &  &{\color[HTML]{9B9B9B}7681}& 2258 & 463& 96 & 14 &  &  &  &  &  &  &  \\
 &sup.&  &  &  &  &{\color[HTML]{9B9B9B}4.0+2.1}& 3.0+1.2 & 2.6+0.8 & 2.3+0.6 & 2.2+0.4 &  &  &  &  &  &  &  \\\hline
\multirow{2}{*}{6} & count &  &  &  &  &  &{\color[HTML]{9B9B9B}504}& 106& 7&  &  &  &  &  &  &  &  \\
 &sup.&  &  &  &  &  &{\color[HTML]{9B9B9B}3.7+1.6}& 3.1+0.9 & 2.7+0.7 &  &  &  &  &  &  &  &  \\\hline
\multirow{2}{*}{7} & count &  &  &  &  &  &  &{\color[HTML]{9B9B9B}3}&  &  &  &  &  &  &  &  &  \\
 &sup.&  &  &  &  &  &  &{\color[HTML]{9B9B9B}4.3+0.5}&  &  &  &  &  &  &  &  & \\\bottomrule
\end{tabular}
}
%}
\end{sidewaystable}

\end{document}